\documentclass[a4paper,english]{lipics-report}

\usepackage{amsfonts, amsmath, amssymb, amsthm}
\usepackage{float}
\usepackage{graphicx}
\usepackage[all]{xy}
\usepackage{mathrsfs}
\usepackage{tikz}
\usetikzlibrary{arrows,shapes,snakes,automata, calc, chains, matrix, positioning, scopes, fit, backgrounds}
\usepackage{xcolor}
\usepackage{xspace}
\usepackage{macros}
\usepackage{authblk}
\usepackage{todonotes}

\newtheorem{claim}[theorem]{Claim}

\begin{document}

\title{Regular and First-order List Functions}

\author[1]{Miko{\l}aj Boja\'nczyk}
\author[2]{Laure Daviaud}
\author[3]{S. Krishna}
\affil[1]{MIMUW, University of Warsaw, Poland}
\affil[2]{DIMAP, Department of Computer Science, University of Warwick, UK}
\affil[3]{Department of Computer Science, IIT Bombay, India}


\maketitle
\abstract

We define two classes of functions,  called \emph{regular (respectively, first-order) list functions}, which manipulate objects such as lists, lists of lists, pairs of lists, lists of pairs of lists, etc. The definition is in the style of regular expressions: the functions are constructed by starting with some basic functions (e.g.~projections from pairs, or head and tail operations on lists) and putting them together using four combinators (most importantly, composition of functions).
Our main results are that first-order list functions are exactly the same as first-order transductions, under a suitable encoding of the inputs; and the regular list functions are exactly the same as  \MSO-transductions. 

\section{Introduction}
\label{section:introduction}
Transducers, i.e.~automata which produce output, are as old as automata themselves, appearing already in Shannon's paper\cite[Section 8]{shannon1948mathematical}. This paper is mainly about string-to-string transducers. Historically, the most studied classes of string-to-string functions were the sequential and rational functions, see e.g.~\cite[Section 4]{Filiot:2016iw} or~\cite{Sakarovitch:2009dk}. Recently, much attention has been devoted to a third,  larger, class of string-to-string functions that we call ``regular'' following~\cite{Engelfriet:2001kv} and~\cite{Alur:2010gc}. The regular string-to-string functions are those recognised by two-way automata~\cite{Aho:1970is}, equivalently by  \MSO transductions~\cite{Engelfriet:2001kv}, equivalently by  streaming string transducers~\cite{Alur:2010gc}. 

In~\cite{Alur:2014gs}, Alur et al.~give yet another characterisation of the regular string-to-string functions,  in the spirit of regular expressions. They identify several basic string-to-string functions, and several ways of combining existing functions to create new ones, in such a way that exactly the regular functions are generated. The goal of this paper is to do the same, but with a different choice of basic functions and combinators. Below we describe some of the differences between our approach and that of~\cite{Alur:2014gs}.

The first distinguishing feature of our approach is that, instead of considering only functions from strings to strings, we allow a richer type system, where functions can manipulate objects such as pairs, lists of pairs, pairs of lists etc. (Importantly, the nesting of types is bounded, which means  that the objects  can be viewed as unranked sibling ordered trees of bounded depth.) This richer type system is a form of syntactic sugar, because the new types can be encoded using strings over a finite alphabet, e.g.~$[(a,[b]),(b,[a,a]), (a,[])]$, and the devices we study are powerful enough to operate on such encodings. Nevertheless, we believe that the richer type system allows us to identify a  more natural and canonical base of functions, with benign   functions such as projection  $\Sigma \times \Gamma \to \Sigma$, or  append $\Sigma \times \Sigma^* \to \Sigma^*$.    Another advantage of the  type system is its tight  connection with programming: since we use standard types and functions on them, our entire set of basic functions and  combinators can be implemented in one screenful of Haskell code,  consisting mainly of giving new names to existing operations. 

A second distinguishing property of our approach is its emphasis on composition of functions. Regular string-to-string functions are closed under composition, and therefore it is natural to add composition of functions as a combinator. However, the system of Alur et al.~is designed so that composition is allowed but not needed to get the completeness result~\cite[Theorem 15]{Alur:2014gs}. In contrast, composition is absolutely essential to our system. We believe that having the ability to simply compose functions -- which is both intuitive and powerful --  is one of the central appeals of transducers, in contrast to other fields of formal language theory  where more convoluted forms of composition are needed, such as wreath products of semigroups  or nesting of languages. With composition, we can leverage deep decomposition results from the algebraic literature (e.g.~the Krohn-Rhodes Theorem or Simon's Factorisation Forest Theorem), and obtain a basis with very simple  atomic operations.
 
Apart from being easily programmable and relying on composition, our system has two other design goals. The first 
goal is that we want it to be easily extensible; we discuss this goal in the conclusions. The second goal is that we want to identify the iteration mechanisms needed for regular string-to-string functions.
 
To better understand the role of iteration, the technical focus of the paper is on  the first-order fragment of regular string-to-string functions~\cite[Section 4.3]{Filiot:2016iw}. Our main technical result, Theorem~\ref{thm:transductions}, shows  a family of atomic functions and combinators that describes exactly  the first-order fragment. We believe that the  first-order fragment is arguably as important as the bigger set of regular functions.  Because of the first-order restriction, some well known sources of iteration, such as modulo counting, are not needed for the first-order fragment. In fact,  one could say that the functions from Theorem~\ref{thm:transductions}  have  no iteration at all  (of course, this can be debated). Nevertheless, despite this lack of iteration, the first-order fragment seems to contain the essence of regular string-to-string functions. In particular, our second main result, which characterises  all regular string-to-string functions in terms of combinators, is obtained by adding product operations for finite groups to the basic functions and then simply applying  Theorem~\ref{thm:transductions} and existing decomposition results from language theory. 

\paragraph*{Organisation of the paper.} In Section~\ref{section:definitions}, we define the class of first-order list functions. In Section~\ref{section:examples}, we give many examples of such functions. 
One of our main results is that the class of first-order list functions is exactly the class of first-order transductions. 
To prove this, we first show in Section~\ref{section:aperiodic} that  first-order list functions contain all the aperiodic rational functions.  Then, in Sections~\ref{section:fotransd} and~\ref{section:register}, we state the result and complete its proof. 
In Section~\ref{section:regular}, we generalise our result to deal with \MSO-transductions. We conclude the paper with future works in Section~\ref{section:conclusion}. 

\section{Definitions}
\label{section:definitions}
We use types that are built starting from finite sets (or even one element sets) and using disjoint unions (co-products), products and lists. More precisely, the set of types we consider is given by the following grammar:
$$\types := \text{every one-element set}  |  \types + \types  |  \types \times \types  |  \types^*$$
For example, starting from elements $a$ of type $\Sigma \in \types$ and $b$ of type $\Gamma \in \types$, one can construct the co-product $\{a, b\}$ of type  $\Sigma + \Gamma$, the product $(a,b)$ of type $\Sigma \times \Gamma$, and the following lists $[a,a,a]$ of type $\Sigma^*$ and $[a,a,b,b,a,a,b]$ of type $(\Sigma + \Gamma)^*$.

For $\Sigma$ in $\types$, we define $\Sigma^+$ to be $\Sigma \times \Sigma^*$.

\subsection{First-order list functions}

The class of functions studied in this paper, which we call \emph{first-order list functions} are functions on the objects defined by the above grammar. It is meant to be large enough to contain natural functions such as projections or head and tail of a list, and yet small enough to have good computational properties (very efficient evaluation, decidable equivalence, etc.). The class is defined by choosing some basic list functions and then applying some combinators.

\begin{definition}[First-order list functions]
\label{def:lf}
Define the \emph{first-order list functions} to be the smallest class of functions having as domain and co-domain any $\Sigma$ from $\types$, which contains all the constant functions, the functions from Figure~\ref{fig:product-coproduct} ($\mathsf{projection}$, $\mathsf{co-projection}$ and $\mathsf{distribute}$), the functions from Figure~\ref{fig:list} ($\mathsf{reverse}$, $\mathsf{flat}$, $\mathsf{append}$, $\mathsf{co-append}$ and $\mathsf{block}$) and which is closed under applying the disjoint union, composition, map and pairing combinators defined in Figure~\ref{fig:combinators}.
\end{definition}

\begin{figure}[htbp]
\small
\fbox{\parbox{\textwidth}{
\begin{align*}
& \bullet \mathsf{\textbf{projection}_1}  && \bullet \bf \mathsf{\textbf{coprojection}}  & & \bullet \bf \mathsf{\textbf{distribute}} \\
& {\Sigma \times \Gamma \longrightarrow \Sigma} & & \Sigma \longrightarrow \Sigma + \Gamma & & (\Sigma + \Gamma) \times \Delta \longrightarrow (\Sigma \times \Delta) + (\Gamma \times \Delta)\\
&  (x,y) \longmapsto x & &  x \longmapsto x & & (x,y) \longmapsto (x,y) 
\end{align*}
}}
\caption{\label{fig:product-coproduct}Basic functions for product and co-product. To avoid clutter, we write only one of the two projection functions but we allow to use both. The types $\Sigma,\Gamma,\Delta$ are from $\types$.}
\end{figure}

\begin{figure}[htbp]
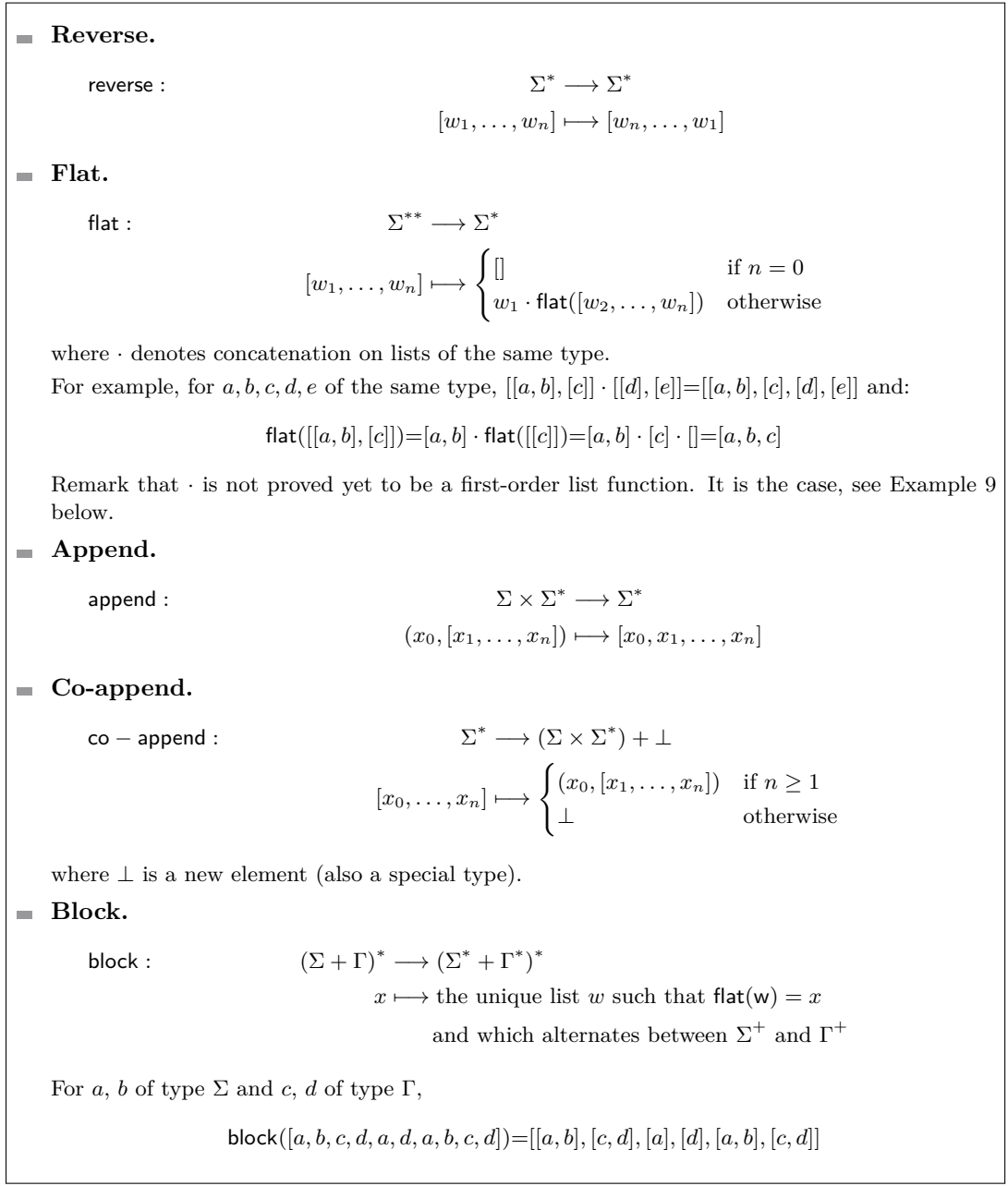

\small
\fbox{\parbox{\textwidth}{
\begin{itemize}
\item {\bf \normalsize Reverse.}
\begin{align*}
& \mathsf{reverse}: & \Sigma^{*} & \longrightarrow \Sigma^* \\
&& [w_1,\ldots,w_n] & \longmapsto [w_n,\ldots,w_1]
\end{align*}

\item {\bf \normalsize Flat.}
\begin{align*}
& \mathsf{flat}: & \Sigma^{**} & \longrightarrow \Sigma^* \\
&& [w_1,\ldots,w_n] & \longmapsto 
\begin{cases}
	[] & \text{if $n=0$}\\
	w_1 \cdot \mathsf{flat}([w_2, \dots, w_n]) & \text{otherwise}
\end{cases}
\end{align*}
where $\cdot$ denotes concatenation on lists of the same type. 

For example, for $a,b,c,d,e$ of the same type, ${[[a,b],[c]] \cdot [[d],[e]]} {=} {[[a,b],[c],[d],[e]]}$ and: 
$$\mathsf{flat}([[a,b],[c]]) {=} {[a,b]\cdot \mathsf{flat}([[c]])} {=} {[a,b]\cdot [c] \cdot []} {=} [a,b,c]$$ 
Remark that $\cdot$ is not proved yet to be a first-order list function. It is the case, see Example~\ref{ex:list-concatenation} below.

\item {\bf \normalsize Append.}
\begin{align*}
& \mathsf{append}: & \Sigma \times \Sigma^* & \longrightarrow \Sigma^* \\
&& (x_0,[x_1,\ldots,x_n]) & \longmapsto [x_0,x_1,\ldots,x_n]
\end{align*}

\item {\bf \normalsize Co-append.}
\begin{align*}
& \mathsf{co-append}: & \Sigma^* & \longrightarrow (\Sigma \times \Sigma^*)+\bot \\
&& [x_0,\ldots,x_n] & \longmapsto 
\begin{cases}
(x_0,[x_1,\ldots,x_n]) & \text{if $n \ge 1$}\\
\bot & \text{otherwise}
\end{cases}
\end{align*}

where $\bot$ is a new element (also a special type).

\item {\bf \normalsize Block.}
\begin{align*}
& \mathsf{block}: & (\Sigma + \Gamma)^* & \longrightarrow (\Sigma^* + \Gamma^*)^* \\
&& x & \longmapsto \text{the unique list $w$ such that $\mathsf{flat(w)}=x$} \\ 
&&& \quad \quad \text{and which alternates between $\Sigma^+$ and $\Gamma^+$}
\end{align*}

For $a$, $b$ of type $\Sigma$ and $c$, $d$ of type $\Gamma$, 
$$\mathsf{block}([a,b,c,d,a,d,a,b,c,d]) {=} [[a,b],[c,d],[a],[d],[a,b],[c,d]]$$
\end{itemize}
}}
\caption{\label{fig:list}Basic functions for lists. The types $\Sigma,\Gamma,\Delta$ are from $\types$.}
\end{figure}

\begin{figure}[htbp]
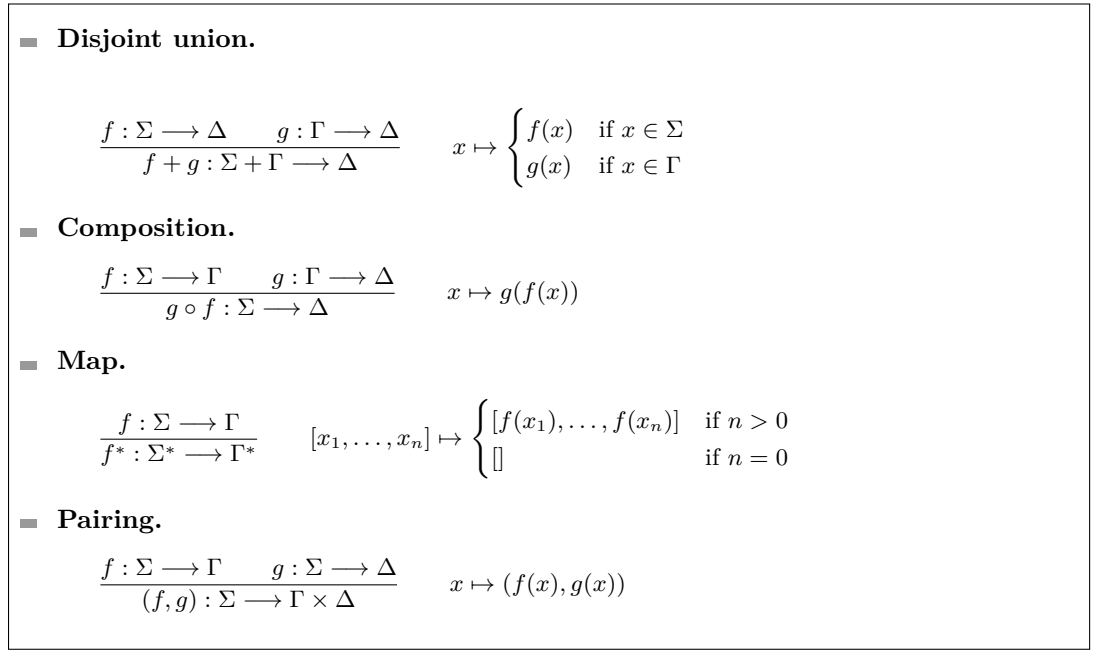

\fbox{\parbox{\textwidth}{
\small
\begin{itemize}
\item {\bf \normalsize Disjoint union.}

\begin{align*}
\combrule{{f: \Sigma \longrightarrow \Delta} \qquad {g: \Gamma \longrightarrow \Delta}}{f+g: \Sigma + \Gamma \longrightarrow \Delta} 
\qquad x \mapsto 
\begin{cases}
  	f(x) & \text{if $x \in \Sigma$}\\
  	g(x) & \text{if $x \in \Gamma$}
\end{cases}
\end{align*}

\item{\bf \normalsize Composition.} 
\begin{align*}
\combrule{{f: \Sigma \longrightarrow \Gamma} \qquad {g: \Gamma \longrightarrow \Delta}}{g \circ f: \Sigma \longrightarrow \Delta} 
\qquad x \mapsto g(f(x))
\end{align*}

\item {\bf \normalsize Map.} 
\begin{align*}
\combrule{{f: \Sigma \longrightarrow \Gamma}}{f^*: \Sigma^* \longrightarrow \Gamma^*} 
\qquad [x_1,\ldots,x_n] \mapsto 
\begin{cases}
  	[f(x_1),\ldots,f(x_n)] & \text{if $n>0$}\\
  	[] & \text{if $n=0$}
\end{cases}
\end{align*}

\item {\bf \normalsize Pairing.} 
\begin{align*}
\combrule{{f: \Sigma \longrightarrow \Gamma} \qquad {g: \Sigma \longrightarrow \Delta}}{(f,g): \Sigma \longrightarrow \Gamma \times \Delta} \qquad x \mapsto (f(x),g(x))
\end{align*}
\end{itemize}
}}
\caption{\label{fig:combinators}Combinators of functions. The types $\Sigma,\Gamma,\Delta$ are from $\types$.}
\end{figure}

\section{Examples}
\label{section:examples}
Natural functions such as identity, functions on finite sets, concatenation of lists, extracting the first element (head), the last element and the tail of a list,... are first order list functions. In this section, we present those examples and prove that they are first-order list functions. Some of these examples  will be used in later constructions.

\begin{myexample}[Identity]
\label{ex:identity}
For every $\Sigma$ in $\types$, the identity function: 
$x \in \Sigma \mapsto x \in \Sigma$ is a first-order list function.
This is achieved by induction on the types: the identity function over a one set element is a constant function and thus a first-order list function. For $\Sigma$ and $\Gamma$ in $\types$, the identity function over $\Sigma + \Gamma$ is the disjoint union of the co-projections $\Sigma \to \Sigma + \Gamma$ and $\Gamma \to \Sigma + \Gamma$. The identity function over $\Sigma \times \Gamma$ is the pairing of the projections $\Sigma \times \Gamma \to \Sigma$ and $\Sigma \times \Gamma \to \Gamma$. Finally, the identity function over $\Sigma^*$ is constructed from the identity function over $\Sigma$ using the Map combinator.
\end{myexample}

\begin{myexample}[List unit function]
\label{ex:unit}
For every $\Sigma$ in $\types$, the function: 
\begin{align*}
x \in \Sigma \mapsto [x] \in \Sigma^*
\end{align*}
is a first-order list function.
This is achieved by first using the pairing combinator which pairs the identity function on $\Sigma$: $f(x)=x$ and the constant function $g$ which maps every $x \in \Sigma$  to the empty list $[]$ of type $\Sigma^*$. 
The pairing combinator gives $(x, [])$ from $x$. This is followed by using the $\mathsf{append}$ function on this pair, giving as result $[x]$.  
\end{myexample}

\begin{myexample}[Functions on finite sets]
\label{ex:finite-functions}
If $\Sigma,\Gamma$ are finite sets, then every function  of type $\Sigma \to \Gamma$ is a first-order list function. This is done by viewing $\Sigma$ as a disjoint union of one-element sets, and then combining  constant functions using the disjoint union combinator.
\end{myexample}

\begin{myexample}[Heads, Last element and Tails]
\label{ex:headtail}
For every $[x_0,x_1, \dots, x_n] \in \Sigma^*$, the function $\mathsf{head}: \Sigma^* \to \Sigma$ which extracts the head $x_0 \in \Sigma$ from the list, the function $\mathsf{last}: \Sigma^* \to \Sigma$ which extracts the last element $x_n \in \Sigma$ from the list as well as the function $\mathsf{tail}: \Sigma^* \to \Sigma^*$ which extracts the tail $[x_1, \dots, x_n]$ from the list are first-order list functions.  To see this, $\mathsf{head}$ is obtained by first using $\mathsf{co-append}$ to  $[x_0,x_1, \dots, x_n]$ obtaining the pair $(x_0, [x_1, \dots, x_n])$ and then projecting out the first component. Likewise, $\mathsf{tail}$ is obtained by projecting the second component. The last element is obtained by reversing the list first, and using $\mathsf{head}$.
\end{myexample}

\begin{myexample}[Length up to a threshold]
\label{ex:length}
For every set $\Sigma$ and $n \in \Nat$, 
\begin{align*}
\mathsf{len}_n : && \Sigma^* & \to \set{0,\ldots,n} \\ && [x_1,\ldots,x_i] & \mapsto 
\begin{cases}
  	i & \text{if $i \le n$} \\ n & \text{otherwise}
\end{cases}
\end{align*}
is a first-order list function.
The proof is by induction on $n$. The function $\mathsf{len}_0: \Sigma^* \to 0$ is constant, and therefore it is a first-order list function. For $n > 0$, we first apply $\mathsf{tail}$. This is composed with the function $\mathsf{len}_{n-1}$
from the induction assumption, and then further composed with the following function:
$x \in \set{0,\ldots,n-1} \mapsto x+1 \in \set{1,\ldots,n}$
which is a first-order list function by Example~\ref{ex:finite-functions}.
For example,  
\begin{align*}
\mathsf{len}_2([x_1,x_2,x_3])=\mathsf{len}_1(\mathsf{tail}([x_1,x_2,x_3]))+1=\mathsf{len}_0(\mathsf{tail}([x_2,x_3]))+1+1=2.
\end{align*}
\end{myexample}

\begin{myexample}[Filter]
\label{ex:filter}
For $\Sigma,\Gamma \in \types$, consider the function:
\begin{align*}
  f : (\Sigma + \Gamma)^* \to \Sigma^*
\end{align*}
which removes the $\Gamma$ elements from the input list. Let us explain why this is a first-order list function. Consider the function from $\Sigma + \Gamma$ to $\Sigma^*$:
\begin{align}\label{eq:unit-or-empty-list}
 a \mapsto \begin{cases}
 	[a] & \text{if $a \in \Sigma$}\\
 	[] & \text{otherwise}
 \end{cases}
\end{align}
which is the disjoint union of the unit function (Example~\ref{ex:unit}) from $\Sigma$ and of the constant function which maps every element of $\Gamma$ to the empty list $[]$ of type $\Sigma^*$. Using map, we apply this function to all the elements of the input list, and then by applying the $\mathsf{flat}$ function, we obtain the desired result.
For example, if $\Sigma=\{a,b,c\}$ and $\Gamma=\{d,e\}$, consider 
the list $[a,c,d,e,b,e,d,a]$ in $(\Sigma+\Gamma)^*$. Using map 
of the above function to this list gives $[[a],[c],[],[],[b],[],[],[a]]$, which after using the function $\mathsf{flat}$, gives $[a,c,b,a]$.
Note that for the types to match, it is important to that the empty list in~\eqref{eq:unit-or-empty-list} is of type $\Sigma^*$.
\end{myexample}

\begin{myexample}[List comma function]
\label{ex:comma-function}
For $\Sigma,\Gamma$ in $\types$, consider the function: 
\begin{align*}
(\Sigma + \Gamma)^* \to \Sigma^{**}
\end{align*}
which groups elements from $\Sigma$ into lists, with elements of $\Gamma$ playing the role of list separators, as in the following example, where $\Sigma=\{a,b,c\}$ and $\Gamma=\{\#\}$:
\begin{align*}
[a,b,\#,c,\#,\#,a,\#,\#,\#,b,c,\#] \mapsto [[a,b],[c],{[]},[a],[],[],[b,c],[]]
\end{align*}
Let us prove that this function is a first order list function. We will freely use the identity function,  the disjoint union and the map combinators without necessarily mentioning it. 
\begin{enumerate}
\item 	We first apply the $\mathsf{block}$ function obtaining: 
\begin{align*}
[[a,b],[\#],[c],[\#,\#],[a],[\#,\#,\#],[b,c],[\#]]
\end{align*}
\item  Next, we  apply $\mathsf{tail}$ to the elements of $\Gamma^{*}$, obtaining: 
\begin{align*}
[[a,b],[],[c],[\#],[a],[\#,\#],[b,c],[]]
\end{align*}
Remark that the empty lists have type $\Gamma^*$.
\item This is followed by applying the list unit function (Example~\ref{ex:unit}) to the elements of $\Sigma^*$ obtaining type $\Sigma^{**}$. 
As a result, we have: 
\begin{align*}
[[[a,b]],[],[[c]],[\#],[[a]],[\#,\#],[[b,c]],[]]
\end{align*}
\item Next, we apply the list unit function (Example~\ref{ex:unit}) to all the elements in a list from $\Gamma^{*}$ (using two nested map) obtaining type $\Gamma^{**}$. This gives: 
\begin{align*}
[[[a,b]],[],[[c]],[[\#]],[[a]],[[\#],[\#]],[[b,c]],[]]
\end{align*}
Remark that now, the empty lists have type $\Gamma^{**}$.
\item Finally, we transform every list of $\Gamma^{**}$ into a list of $\Sigma^{**}$, by mapping every non empty list of $\Gamma^{*}$ to the empty list of $\Sigma^{*}$. The example gives:
\begin{align*}
[[[a,b]],[],[[c]],[[]],[[a]],[[],[]],[[b,c]],[]]
\end{align*}
Remark now that all the elements are of type $\Sigma^{**}$.   
\item Finally, if the first and last elements are the empty list of $\Sigma^{**}$, they are mapped to $[[]]$ of $\Sigma^{**}$, by using $\mathsf{head}$, $\mathsf{last}$, $\mathsf{append}$, $\mathsf{co-append}$, the pairing combinator and reversing the list. The example gives then: \begin{align*}
[[[a,b]],[],[[c]],[[]],[[a]],[[],[]],[[b,c]],[[]]]
\end{align*}
\end{enumerate}
It is now sufficient to use $\mathsf{flat}$ to obtain the desired result.
\end{myexample}

\begin{myexample}[Pair to list] 
\label{ex:pair-to-list}
We can convert a pair to list of length two as follows: use the pairing combinator on the two following functions:
\begin{align*}
& (x,y) \mapsto x && \text{ obtained using } \mathsf{projection}_1 \\
\text{and } & (x,y) \mapsto [y] && \text{ obtained by composition of } \mathsf{projection}_2 \\ &&& \text{ and the list unit function (Example~\ref{ex:unit})}
\end{align*}
in order to get $(x,[y])$. Then, $\mathsf{append}$ gives  $[x,y]$.

To get the converse translation, we use first $\mathsf{co-append}$ on $[x,y]$ to get $(x,[y])$. This is followed by  pairing  the two functions:
\begin{align*}
& (x,[y]) \mapsto x && \text{ obtained using } \mathsf{projection}_1 \\
\text{and } & (x,[y]) \mapsto y && \text{ obtained by composition of } \mathsf{projection}_2 \\ &&& \text{ and the function $\mathsf{head}$ (Example~\ref{ex:headtail})}
\end{align*}
to get $(x,y)$.

Note that for the second translation, the type of the output is 
  $(\Sigma \times (\Sigma + \bot)) + \bot$.
If we want the output type to be $\Sigma \times \Sigma$, then we can choose some element $c \in \Sigma$ and send the first $\bot$ to $c$ and the second $\bot$ to $(c,c)$. In this case, the resulting function will satisfy:
\begin{align*}
  [x_1,\ldots,x_n] \mapsto \begin{cases}
  	(c,c) & \text{for $n=0$}\\
  	(x_1,c) & \text{for $n=1$} \\
  	(x_1,x_2) & \text{otherwise}.
  \end{cases}
\end{align*}
\end{myexample}

\begin{myexample}[List concatenation]
\label{ex:list-concatenation}
List concatenation is a first-order list function: consider $([x_1,\ldots,x_n]$,$[y_1,\ldots,y_k])$, 
 a pair of lists and apply Example~\ref{ex:pair-to-list} to obtain $[[x_1,\ldots,x_n],[y_1,\ldots,y_k]]$
and then $\mathsf{flat}$ to get: $[x_1,\ldots,x_n,y_1,\ldots,y_k]$.
\end{myexample}

\begin{myexample}[Windows of size 2]
\label{ex:window}
For every $\Sigma$ in $\types$, the following is a first-order list function:
\begin{align*}
[x_1,\ldots,x_n] \in \Sigma^*  \mapsto  [(x_1,x_2),(x_2,x_3),\ldots,(x_{n-1},x_n)] {\in} (\Sigma \times \Sigma)^*
\end{align*}
(When the input has length at most 1, then the output is empty.) Let us show how to get the above function. Consider a list $[x_1,\ldots,x_n]$.
Using the same ideas as in Example~\ref{ex:unit}, the following function is a first-order list function:
\begin{align*}
  x \in \Sigma  \qquad \mapsto \qquad [x,\#,x] \in (\Sigma + \{\#\})^*.
\end{align*}
Apply the above function to every element of the input list, and then use $\mathsf{flat}$ on the result, yielding a list of type $(\Sigma + \{\#\})^*$ of the form:
\begin{align*}
  [x_1,\#,x_1,x_2,\#,x_2,\ldots,x_n,\#,x_n]
\end{align*}
Apply the list comma function from Example~\ref{ex:comma-function}, to get a list of the form:
\begin{align*}
  [[x_1],[x_1,x_2],[x_2,x_3],\ldots,[x_{n-1},x_n],[x_n]]
\end{align*}
Remove the first and last elements (using $\mathsf{head}$ and $\mathsf{last}$), yielding a list of the form:
\begin{align*}
  [[x_1,x_2],[x_2,x_3],\ldots,[x_{n-1},x_n]].
\end{align*}
Finally, apply the function from Example~\ref{ex:pair-to-list} to each element, yielding the desired list:
\begin{align*}
  [(x_1,x_2),(x_2,x_3),\ldots,(x_{n-1},x_n)].
\end{align*}
The ideas in this example can be extended to produce windows of size 3,4, etc.
\end{myexample}

\begin{myexample}[If then else]
\label{ex:if-then-else} 
Suppose that $f : \Sigma  \to \set{0,1}$ and $g_0,g_1 : \Sigma \to \Gamma$ are first-order list functions. Then $x \mapsto g_{f(x)}(x)$ is also a first-order list function. This is done as follows. On input $x \in \Sigma$, we first apply the pairing of $f$ and the identity function, yielding a result:
\begin{align*}
  (f(x),x) \in \set{0,1} \times \Sigma.
\end{align*}
Next we apply the function $\mathsf{distribute}$, transforming the type into:
\begin{align*}
  (f(x),x) \in \left(\set{0} \times \Sigma\right) + \left(\set{1} \times \Sigma\right)
\end{align*}
To this result we apply the disjoint union $h_0 + h_1$ where $h_i$ is defined by $(i,y) \mapsto g_i(y)$, yielding the desired result.
\end{myexample}

\begin{myexample}
Every function $f : \Sigma \to \Gamma$ can be lifted to a function $f^+ : \Sigma^+ \to \Gamma^+$ in the natural way, and  first-order list functions are  easily seen to be closed under this lifting by using the map and pairing combinators.
\end{myexample}

\section{Aperiodic rational functions}
\label{section:aperiodic}
The main  result of this section is  that  the class of first-order list functions contains all aperiodic  rational functions, see~\cite[Section IV.1]{DBLP:books/daglib/0023547} or Definition~\ref{def:aperiodic-rational-function} below. An important part of the proof is that first-order list functions can compute factorisations as in the Factorisation Forest Theorem of Imre Simon~\cite{Simon1990,Bojanczyk2009}. In Section~\ref{sec:simon}, we state that the Factorisation Forest Theorem can be made effective using first-order list functions, and in Section~\ref{sec:sequential}, we define aperiodic rational functions and prove that they are first-order list functions. 

\subsection{Computing factorisations}
\label{sec:simon}

In this section we state the Factorisation Forest Theorem and show how it  can be made effective using first-order list functions. We begin by defining monoids and semigroups. For our application, it will be convenient to use a definition where  the product operation is not  binary, but has unlimited arity. (This is the view of monoids and semigroups as Eilenberg-Moore algebras over monads $\Sigma^*$ and $\Sigma^+$, respectively).

\begin{definition}
\label{def:monoid}
A \emph{monoid} consists of a set $M$ and a product operation $\pi : M^* \to M$ which is associative, ie for all elements $m_1, \ldots m_k$ of $M$ and for all $1 \leq \ell_1 < \ell_2 < \dots < \ell_j < k$:
\begin{align*}
\pi(m_1, \ldots, m_k) = \pi(\pi(m_1, \ldots, m_{\ell_1}),\pi(m_{\ell_1+1}, \ldots, m_{\ell_2}),\ldots,\pi(m_{\ell_j+1}, \ldots, m_k))
\end{align*}
\end{definition}
	
Remark that by definition the empty list of $M^*$ is sent by $\pi$ to a neutral element in $M$. 
A \emph{semigroup} is defined the same way, except that nonempty lists $M^+$  are used instead of possibly empty ones. 

\begin{definition}
A monoid (or semigroup) is called \emph{aperiodic} if there exists some positive integer $n$ such that 
$m^n = m^{n+1}$ for every element $m \in M$
where $m^n$ denotes the $n$-fold product of $m$ with itself.		
\end{definition}

A \emph{semigroup homomorphism} is a function between two semigroups which is compatible with the semigroup product operation.

\paragraph*{Factorisations.}
Let $h : \Sigma^+ \to S$ be a semigroup homomorphism (equivalently, $h$ can be given as a function $\Sigma \to S$ and extended uniquely into a homomorphism). An \emph{$h$-factorisation} is defined to be a sibling-ordered tree which satisfies the following constraints (depicted in the following picture):
leaves are labelled by elements of $\Sigma$ and have no siblings. All the other nodes are labelled by elements from $S$. The parent of a leaf labelled by $a$ is labelled by $h(a)$. The other nodes have at least two children and are labelled by the product of the child labels. If a node has at least three children then those children have all the same label. 

\mypicsmall{3}

\paragraph*{Computing factorisations using first-order list functions.} As described above, an $h$-factorisation is a special case of a tree  where leaves have labels in $\Sigma$ and non-leaves have labels in $S$. Objects of this type, assuming that there is some bound $k$ on the depth, can be represented using our type system:
\begin{eqnarray*}
  \trees_0(\Sigma, S) &=& \Sigma \\
   \trees_{k+1}(\Sigma,S)&=& \trees_k(\Sigma,S) +   S \times (\trees_k(\Sigma,S))^+
\end{eqnarray*}
Using the above representation, it is meaningful to talk about a first-order list function computing an $h$-factorisation of depth bounded by some constant $k$. This is the representation used in the following theorem. The theorem is essentially the same as the Factorisation Forest Theorem (in the aperiodic case), except that it additionally says that the factorisations can be produced using first-order list functions. 

\begin{theorem}
\label{thm:simon-compute}
Let $\Sigma \in \types$ be a (not necessarily finite) type and let $h : \Sigma \to S$ be a function into the universe of some finite aperiodic semigroup $S$. If $h$ is a first-order list function then there is some $k \in \Nat$ and a first-order list function:
\begin{align*}
  f : \Sigma^+ \to \trees_k (\Sigma,S)
\end{align*}
 such that for every $w \in \Sigma^+$, $f(w)$ is an $h$-factorisation whose yield (i.e.~the sequence of leaves read from left to right) is $w$.
\end{theorem}

Before giving the proof of this theorem, let us give a corollary.

\begin{corollary}
\label{cor:products} 
Let $\Sigma \in \types$ be finite. Then the following functions are first-order list functions:
\begin{enumerate}
	\item  every semigroup homomorphism $h : \Sigma^+ \to S$ where  $S$ is finite aperiodic;
	\item every regular language over $\Sigma$, viewed as a function $\Sigma^* \to \set{0,1}$.
\end{enumerate}
\end{corollary}

\begin{pr}[of the corollary~\ref{cor:products}]
For item 1, because $\Sigma$ is finite, $h : \Sigma \to S$ is a first-order list function. We can then use Theorem~\ref{thm:simon-compute}, compute a $h$-factorisation with a first-order list function and then output its root label.	 
For item 2, take some semigroup homomorphism which recognises the language, apply item 1, and compose with the characteristic function of the accepting set. That is, consider the homomorphism $h: \Sigma^+ \rightarrow S$ recognizing $L \subseteq \Sigma^+$ and $P=h(L)$.
Compose this with $g: S \rightarrow \{0,1\}$ which maps exactly the elements of $P$ to 1 (which is a first order list function because $S$ is finite). 
Since the composition of two first-order  list functions is a first-order list function, we 
obtain item 2. However,  in item 2, we need to treat  separately the case of the empty list on input, but this can be done using Examples~\ref{ex:if-then-else} and~\ref{ex:length}.
\end{pr}

\begin{pr}[of Theorem~\ref{thm:simon-compute}]
Suppose that $s_1,\ldots,s_n$ are elements of $S$. In the proof below, we adopt the notational convention that $[s_1,\ldots,s_n] \in S^+$ represents the list of these elements, while $s_1 \cdots s_n \in S$ represents their product. In particular, $st$ denotes the element  of $S$ which is the product of two elements $s$ and $t$.

The proof of the theorem is by induction on the following parameters: (a) the size of $S$; and (b) the size of the image $h(\Sigma)$.  These parameters are ordered lexicographically, i.e.~we can call the induction assumption for a smaller semigroup even if the size of $h(\Sigma)$ grows.
	
\begin{enumerate}
\item Consider first the induction base, when the set $h(\Sigma)$ contains only one element, call it $s$. First, using the function from Example~\ref{ex:length}, we show that 
\begin{align*}
  [a_1,\ldots,a_n] \in \Sigma^+ \quad  \mapsto  \quad \overbrace{s \cdots s}^{\text{$n$ times}} \in S
\end{align*}
are first-order list functions. The key observation is that, since $S$ is aperiodic, the above function is constant for lists whose length exceeds some threshold. Pairing the above function with:
\begin{align*}
  [a_1,\ldots,a_n] \in \Sigma^+ \mapsto  [(s,a_1),\ldots,(s,a_n)] \in (S \times \Sigma)^+
\end{align*}
we get the conclusion of the theorem. 

\item Suppose that there is some $s \in h(\Sigma)$ such that:
\begin{align*}
T =   \set { ts : t \in S} \varsubsetneq S
\end{align*}
is a proper subset of $S$.  Note that $T$ is a subsemigroup of $S$. 
Consider  two copies of $\Sigma$, $\redSigma$ the red copy and $\blueSigma$ the blue one, and the function:
 $f : \Sigma \to \redSigma + \blueSigma$
which colours red those elements of $\Sigma$ which are mapped with $s$, and colours blue the remaining ones (formally speaking, the range of $f$ is a co-product of two copies of $\Sigma$). This is a first-order list function, by using the if-then-else construction described in Example~\ref{ex:if-then-else}. To an input list in $\Sigma^*$, apply $f$ to all the elements of the list (with map), and then apply the $\mathsf{block}$ function, yielding a list:
 $x \in  (\redSigma^* + \blueSigma^*)^*$. 
Assume first that $x$ begins with a blue list and ends with a red list and has the form 
\begin{align*}
&  [\bluez 1,\redy 1,\bluez 2,\redy 2,\ldots,\bluez n,\redy n], \\
&  \bluez 1,\ldots,\bluez n \in \blueSigma^+, \redy 1,\ldots,\redy n \in \redSigma^+
\end{align*}
Using the window function from Example~\ref{ex:window} and discarding pairs that are of type $\redSigma^+ \times \blueSigma^+$ with the filtering function from Example~\ref{ex:filter}, we can transform the above list into one of the form:
\begin{align*}
  [(\bluez 1,\redy 1),(\bluez 2,\redy 2),\ldots,(\bluez n,\redy n)].
\end{align*}
To both $\redSigma^*$ and $\blueSigma^*$ we can apply the induction assumption on the number of generators. Therefore, using the induction assumption, co-product and map, we can transform the above list into a list of $h$-factorisations: 
\begin{align*}
  [u_1,u_2,\ldots,u_n]
\end{align*}
such that each $u_i$ is an $h$-factorisation of $\bluez i \redy i$. The key observation is that, since $\redy i$ is a nonempty list of elements with value $s$, it follows that the value of $u_i$ in the semigroup belongs to the set $T$, which is a smaller semigroup than $S$. Therefore, we can apply the induction assumption again, to transform the above list into an $h$-factorisation.
We can treat similarly the cases where $x$ does not begin with a blue list or does not end with a red list.

\item If there is some $s \in h(\Sigma)$ such that
$T =   \set { st : t \in S}$ 
is a proper subset of $S$, then we proceed analogously as in the previous case.

\item We claim that one of the above three cases must hold. Indeed, if neither case 2 nor 3 holds,  then the functions:
\begin{align*}
  f_s: t \mapsto st \qquad \text{and} \qquad g_s: t \mapsto ts
\end{align*}
are permutations of $S$ for all $s$. In particular, by the assumption that $S$ is aperiodic, we deduce that $f_s$ and $g_s$ are the identity on the semigroup generated by $s$ and $s^2 = s$. Thus for all $t$, $st=s^2t$ and then necessarily $f_s$ is the identity over $S$. The same thing holds for $g_s$. Therefore for every $s,t \in h(\Sigma)$ we have $t =   st = s$. 
This means we are in case 1.
\end{enumerate}
\end{pr}

\subsection{Rational functions}
\label{sec:sequential}
For the purposes of this paper, it will be convenient to give an algebraic representation for rational functions. In this section, we will only be interested in the case of aperiodic ones; however we explain in section~\ref{section:regular} that our results can be generalised to arbitrary rational functions.

\begin{definition}[Rational function]
\label{def:aperiodic-rational-function}
	The syntax of a \emph{rational function} is given by:
	\begin{itemize}
		\item input and output alphabets $\Sigma,\Gamma$, which are both finite;
		\item a monoid homomorphism $h : \Sigma^* \to M$ with $M$ a finite monoid;
		\item an output function $out : M \times \Sigma \times M \to \Gamma^*$.
	\end{itemize}
		If the monoid $M$ is aperiodic, then the rational function is also called \emph{aperiodic}.
	The semantics is the function:
\begin{align*}
  a_1 \cdots a_n \in \Sigma^*  \quad \mapsto \quad w_1 \cdots w_n  \in \Gamma^*
\end{align*}
where $w_i$ is defined to be the value of the output function on the triple:
\begin{enumerate}
	\item value under $h$ of the prefix  $a_1 \cdots a_{i-1}$
	\item letter $a_i$
	\item value under $h$ of the suffix $a_{i+1} \cdots a_n$.
\end{enumerate}
\end{definition}

Note that in particular, the empty input word is mapped to an empty output. 

\begin{theorem}\label{thm:sequential-functions}
Every aperiodic rational function  is a first-order list function.
\end{theorem}

The rest of Section~\ref{sec:sequential} is devoted to showing the above theorem. The general idea is to use factorisations as in Theorem~\ref{thm:simon-compute} to compute the rational function.
 
\paragraph*{Sibling profiles.}
Let $h : \Sigma^* \to M$ be a homomorphism into some finite aperiodic monoid $M$. Consider an $h$-factorisation, as defined in Section~\ref{sec:simon}. For a non-leaf node $x$ in the $h$-factorisation, define its \emph{sibling profile} (see Figure~\ref{fig:siblings}) to be the pair $(s,t)$ where $s$ is the product in the monoid of the labels in the left siblings of $x$, and $t$ is the product in the monoid of the labels in the right siblings. If $x$ has no left siblings, then $s=1$, if $x$ has no right siblings then $t=1$ (where $1$ denotes the neutral element of the monoid $M$).

\begin{figure}[hbt]
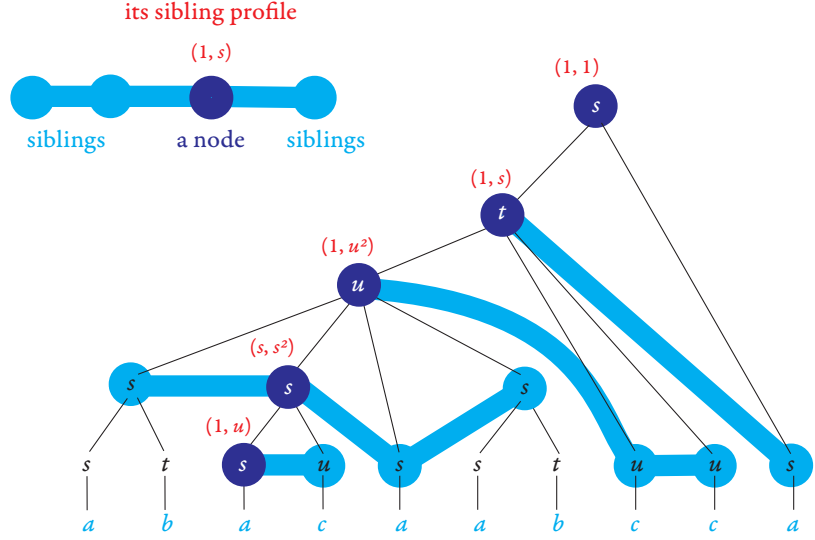

\mypic{8}
  \caption{\label{fig:siblings}Sibling profiles.}
\end{figure}

The two following lemmas give transformations on trees (which are $h$-factorisation) that are first-order list functions.

\begin{lemma}
\label{lem:sibling-profiles}
Let $k \in \Nat$ and $h$ be a homomorphism $\Sigma^* \to M$. There is a first-order list function
  $\trees_k (M,\Sigma) \to \trees_k (M \times M, \Sigma)$
which transforms any $h$-factorisation 
by replacing the label of each non-leaf node with its sibling profile.
\end{lemma}

\begin{pr}
We prove the lemma by induction on $k$. We use the following claim to deal with nodes of degree at least 3 in the induction step.
\begin{claim}\label{claim:degreethree}
Let  $\Delta \in \types$ and let $s \in M$. The function which maps a list $[x_1,\ldots,x_n] \in \Delta^*$ to: 
\begin{align*}
& [((s^0,s^{n-1}),x_1),\ldots,((s^{i-1},s^{n-i}),x_i),\ldots,((s^{n-1},s^0),x_n)] \in (M \times \Delta \times M)^*
\end{align*} 
is a first-order list function.
\end{claim}
\begin{pr}
Since $M$ is aperiodic, there is some $n_0$ such that all powers $s^n$ with $n > n_0$ are the same. We use $\mathsf{tail}$ ($n_0$ times), $\mathsf{reverse}$ and again $\mathsf{tail}$ ($n_0$ times) to extract the list consisting of elements that are at distance at least $n_0$ to both the beginning and end of the list, and apply the function $x \mapsto ((s^{n_0},s^{n_0}),x)$ to all those elements. We treat one by one the remaining elements, i.e.~those at distance at most $n_0$ from either the beginning or end of the list, because there is at most $2n_0$ such elements, and we can extract them using $\mathsf{head}$ and $\mathsf{last}$ at most $n_0$ times.
\end{pr}

We can now give the inductive proof. For $k=0$, the identity function, which is  a first-order list function satisfies the conditions. Let $k>0$. If the root has degree at most $2$, let us write $t_1$ and $t_2$ for its two subtrees and $s_1$ and $s_2$ for the label of their respective roots. By induction, there exist a first-order list function transforming $t_i$ into $t'_i$, for $i=1,2$, where each node (except the root) is replaced by its sibling profile. One can compose it with a first-order list function which replaces the node corresponding to $s_1$ by $(1,s_2)$ and $s_2$ by $(s_1,1)$.
If the root has degree at least $3$, the reasoning is similar. Let us write $t_1,\ldots, t_n$ for the subtrees and $s$ for the label of their respective roots. By induction, there exist a first-order list function transforming $t_i$ into $t'_i$, for $i=1,\ldots,n$, where each node (except the root) is replaced by its sibling profile. We can now use the function from Claim~\ref{claim:degreethree}, to replace the label of the roots of the subtrees by their sibling profile.
\end{pr}

\begin{lemma}
\label{lem:ancestor-labels-recursive}
Let $k \in \Nat$ and let $\Delta$ be a finite set. Then there is a first-order list function: 
  $\trees_k (\Delta,\Sigma) \to ( \Delta^* \times \Sigma)^*$
which inputs a tree and outputs the following list: for each leaf (in left-to-right order) output the label of the leaf plus the sequence of labels in its ancestors listed in increasing order of depth.
\end{lemma}

\begin{pr}
The key assumption here is that the depth of trees is bounded. We will prove the lemma by induction. For $k=0$, the function $a \in \Sigma \mapsto [([],a)]$ which is a first-order list functions satisfies the conditions in the statement. Let $k>0$, and $(s,[t_1,\ldots ,t_n]) \in \trees_k (\Delta,\Sigma)$. By induction hypothesis, there is a first-order list function $f$ transforming all the $t_i$ into a list as stated in the lemma. Using map (and pairing with identity), we can apply this function to all the subtrees $t_1,\ldots, t_n$ and get a pair $(s,[f(t_1),\ldots ,f(t_n)])$. We then just need to concatenate $s$ (which we can extract from the pair using projection) to all the lists of the ancestors already paired with the leaves, which we can do using projection, map and append. We get a pair $(s,\ell)$ where $\ell$ is a list of lists of pairs (list of ancestors, leaf). We finally project on the second element and flatten to get the desired list.
\end{pr}

\begin{pr}[of Theorem~\ref{thm:sequential-functions}]
Let $r : \Sigma^* \to \Gamma^*$ be a rational function, whose syntax is given by
 $ h : \Sigma^* \to M \quad \text{and} \quad out : M \times \Sigma \times M \to \Gamma^*$. 
Our goal is to show that $r$ is a first-order list function. We will only show how to compute $r$ on non-empty inputs. To extend it to the empty input we can use an if-then-else construction as in Example~\ref{ex:if-then-else}. We will define $r$ as a composition of five functions, described below. To illustrate these steps, we will show after each step the intermediate output, assuming that the input is a word from $\Sigma^+$ that looks like this: \mypicsmall{10}

\begin{enumerate}
\item Apply Theorem~\ref{thm:simon-compute} to $h$, yielding some $k$ and a function:
 $\Sigma^+ \to \trees_k(M,\Sigma)$ 
which maps each input to an $h$-factorisation. After applying this function to our input, the result is an $h$-factorisation which looks like:
\mypicsmall{12}

\item To the $h$-factorisation produced in the previous step, apply the function from Lemma~\ref{lem:sibling-profiles}, which replaces the label of each non-leaf node with its sibling profile. After this step, the output looks like this: \mypicsmall{11}

\item To the output from the previous step, we can now apply the function from Lemma~\ref{lem:ancestor-labels-recursive}, pushing all the information to the leaves, so that the output is a list that looks like this: \mypicsmall{9}

\item For $k$ as in the first step, consider the function
$g:    (M \times M)^* \times \Sigma \to M \times \Sigma \times M + \bot$
defined by:
 \begin{align*}
  & ([(s_1,t_1),\ldots,(s_n,t_n)],a) \quad \mapsto \quad 
  \begin{cases}
  	(s_1 \cdots s_n, a, t_n \cdots t_1) & \text{if $n \le k$} \\
  	\bot & \text{otherwise}
  \end{cases}
\end{align*}
The  function $g$ is a first-order list function, because it returns $\bot$ on all but finitely many arguments. Apply $g$ to all the elements of the list produced in the previous step (using map), yielding a list from $(M \times \Sigma \times M)^*$ which looks like this: 
\mypicsmall{13}

\item In the list produced in the previous step, the $i$-th position stores the $i$-th triple as in the definition of rational functions (Definition~\ref{def:aperiodic-rational-function}). Therefore, in order to get the output of our original rational function $r$, it suffices to apply $out$ (because of finiteness, $out$ is a first-order list function) to all the elements of the list obtained in the previous step (with map), and then use $\mathsf{flat}$ on the result obtained.
\end{enumerate}
\end{pr}

\section{First-order transductions}
\label{section:fotransd}
This section states the main result of  this paper:  the  first-order list  functions  are exactly those that can be defined using first-order transductions (\FO-transductions).  We begin by describing \FO-transductions in Section~\ref{subsection:fotransdef}, and then in  Section~\ref{subsection:listasstructure}, we show how they can be applied to types from $\types$ by using an encoding of lists, pairs, etc.~as  logical structures; this allows us to state our main result, Theorem~\ref{thm:transductions}, namely that \FO-transductions are the same as first-order list functions (Section~\ref{subsection:mainresult}). The proof of the main result is given in Sections~\ref{subsection:mainresult}, \ref{subsection:sst} and~\ref{section:register}.

\subsection{\FO-transductions: definition}
\label{subsection:fotransdef}
A \emph{vocabulary} is a set (in our application, finite) of relation names, each one with an associated arity (a natural number). We do not use functions. If $\Vv$ is a vocabulary, then a \emph{logical structure} over $\Vv$ consists of a universe (a set of elements), together with an interpretation of each relation name in $\Vv$ as a relation on the universe of corresponding arity. 

An \emph{\FO-transduction} \cite{CourcelleBook} is a method of transforming one logical structure into another which is described in terms of first-order formulas. More precisely, an \FO-transduction consists of two consecutive operations: first, one copies the input structure a fixed number of times, and next, one defines the output structure using a first-order interpretation (we use here  what is sometimes known as a \emph{one dimensional interpretation}, i.e.~we cannot use pairs or triples of input elements to encode output elments). The formal definitions are given below.

\paragraph*{One dimensional \FO-interpretation.}
The syntax of a one dimensional \FO-interpretation (see also~\cite[Section 5.4]{hodges1993model}) consists of:
\begin{enumerate}
	\item Two  vocabularies, called the \emph{input} and \emph{output} vocabularies.
 	\item A formula  of first-order logic with one free variable over the input vocabulary, called the \emph{universe formula}.
 	\item For each  relation name $R$ in the output vocabulary, a formula $\varphi_R$ of first-order logic over the input vocabulary, whose number of free variables is equal to the arity of $R$. 
\end{enumerate}
The semantics   is a function from logical structures over the input vocabulary to logical structures over the output vocabulary given as follows. The universe of the output structure consists of those elements in the universe of the input structure which make the universe formula true. A predicate $R$ in the output structure  is interpreted as those tuples which are in the universe of the output structure and make the formula $\varphi_R$ true. 
\paragraph*{Copying.}
For a positive integer $k$ and a vocabulary $\Vv$, we define $k$-copying (over $\Vv$) to be the function which inputs a logical structure over $\Vv$, and outputs $k$ disjoint copies of it, extended  with an additional $k$-ary predicate  that selects a tuple $(a_1,\ldots,a_k)$  if and only if there is some $a$ in the input structure such that $a_1,\ldots,a_k$ are the respective copies  of $a$. (The additional predicate is sensitive to the ordering of arguments, because we distinguish between the first copy, the second copy, etc.) 

\begin{definition}[\FO-transduction]
\label{def:fo-transduction}
An \FO-transduction is defined to be an operation on relational structures which is the composition of $k$-copying for some $k$, and of a one dimensional \FO-interpretation. 
\end{definition}

\FO-transductions are a robust class of functions. In particular, they are closed under composition. 
Perhaps even better known are the more general \MSO-transductions, we will discuss these at the end of the paper.

\paragraph*{An example of  \FO-transduction.} We give here a simple example of an \FO-transduction.  Consider the word structure 
\begin{align*}
\Ss=(\Uu, S, <, (Q_a)_{a \in \Sigma})
\end{align*}
over a finite alphabet $\Sigma=\{a,b\}$. The universe $\Uu$ is a finite set of positions $\{0, 1, \dots, n\}$ in the word.
For first order variables $x, y$, we have 
the relations $x < y$ and the successor relation $S(x,y)$ 
with the obvious meanings. We also 
have the relation 
$Q_a(x)$ which evaluates to true if $x$ can be assigned some value 
$i \in \Uu$ such that 
the $i$th position 
of the word has an $a$. For example, the word 
$ababa$ satisfies the formula 
\begin{align*}
& \exists x[first(x) \wedge Q_a(x)] \\
\wedge & \exists x[last(x) \wedge Q_a(x)] \\
\wedge & \forall x, y [(S(x,y) \wedge Q_a(x) \rightarrow \neg Q_a(y)) \wedge (S(x,y) \wedge Q_b(x) \rightarrow \neg Q_b(y)]
\end{align*}
where 
$first(x)=\forall y(x \leq y)$ and $last(x)=\forall y(y \leq x)$. 

Consider the transduction which transforms a word $w$ into $w_1w_2$ where $w_1$ and $w_2$ respectively 
are obtained by removing the $b$'s and $a$'s from $w$. For example $ababa$ is transformed into $aaabb$.
\begin{enumerate}
\item 	
We make two copies of the input structure. The nodes in the first copy labelled by an $a$ as well as the nodes in the second copy labelled by a $b$ are in the universe of the output word. They are specified by the \FO-formula $\varphi^1(x)=Q_a(x)$ and $\varphi^2(x)=Q_b(x)$. 
\item The edges between nodes in the first copy are specified by the 
formula  
\begin{align*}
\varphi^{1,1}(x,y)=x< y \wedge \neg \exists z(x < z< y \wedge Q_a(z))
\end{align*} 
which allows an edge between 
an $a$ and the next occurrence of an $a$. Likewise, edges between nodes in the second  copy are specified by the 
formula  
\begin{align*}
\varphi^{2,2}(x,y)=x< y \wedge \neg \exists z(x < z< y \wedge Q_b(z))
\end{align*} 
\item Finally, we specify edges between the nodes of copy 1 and copy 2. 
The formula 
\begin{align*}
\varphi^{2,1}(x,y)=false
\end{align*}
disallows any edges from the second copy to the first copy, while 
the formula 
\begin{align*}
\varphi^{1,2}(x,y)=(Q_a(x) \wedge \forall z(z > x \rightarrow \neg Q_a(z))) \wedge 
(Q_b(y) \wedge \forall z(z < y \rightarrow \neg Q_b(z)))
\end{align*}
enables an edge from the last $a$ (the last position in the first copy)
to the first $b$ (the first position in the second copy).  
\end{enumerate}

This results in the word where all the $a$'s in $w$ appear before all the $b$'s in $w$. 

\subsection{Nested lists as logical structures}
\label{subsection:listasstructure}

Our goal is to use \FO-transductions to define functions of the form $f : \Sigma \to \Gamma$, for types $\Sigma, \Gamma \in \types$. To do this, we need to represent elements of $\Sigma$ and $\Gamma$ as logical structures. We use a natural encoding, which is essentially the same one as is used in the automata and logic literature, see e.g.~\cite[Section 2.1]{Thomas1997}.

Consider a type   $\Sigma \in \types$.  We represent an element $x \in \Sigma$ as a relational structure, denoted by $\underline x$, as follows:
\begin{enumerate}
	\item The universe $\Uu$ is the nodes in the parse tree of $x$ (see Figure~\ref{fig:parse-tree}).  
	\item There is a binary  relation $\pare(x,y)$ for the parent-child relation which says that $x$ is the parent of $y$. 
	\item There is a binary relation $\sib(x,y)$ for the transitive closure of the ``next sibling'' relation.  
	The next sibling relation $\nsib(x,y)$ is true if $y$ is the next sibling of $x$: 
	that is,  there is a node $z$ which is the parent 
	of $x$ and $y$, and there are no children of $z$ between $x$ and $y$ (in that order).  
	$\sib(x,y)$ evaluates to true if $x, y$ are siblings, and $y$ after $x$. 
		\item   For every node $\tau$ in the parse tree of the type $\Sigma$ (see Figure~\ref{fig:syntax-tree}), there  is a unary predicate $\type(\tau)$, which selects the elements from the universe of $\underline x$, (equivalently the subterms of $x$) 
		that have the type as $\tau$. For example, for the node $\tau$ labeled with $[b]$, 
		$B^*(\tau)$ evaluates to true if $b \in B$.  
\end{enumerate}

\begin{figure}[hbt]
\begin{center}
\begin{tikzpicture}[sibling distance=3.5em, level distance=4em]
\tikzstyle{one}=[]
\node[one] (f) {$([[a,b],[a,a,b],[],c],[(a,[b])])$}
    child { node[one] {$[[a,b],[a,a,b],[],c]$} 
    	child { node[one] {$[a,b]$} 
    		child { node[one] {$a$} } 
    		child { node[one] {$b$} } 
    		child { node[one] {} edge from parent[draw=none] } } 
    	child { node[one] {$[a,a,b]$} 
    		child { node[one] {} edge from parent[draw=none] }
    		child { node[one] {$a$} } 
    		child { node[one] {$a$} } 
    		child { node[one] {$b$} } }
    	child { node[one] (a) {$[]$} }
    	child { node[one] (b) {$c$} } 
    	child { node[one] {} edge from parent[draw=none] } }	
    child { node[one] {} edge from parent[draw=none] } 		
	child { node[one] (h) {$[(a,[b])]$}
	    child { node[one] {} edge from parent[draw=none] } 
		child { node[one] {$(a,[b])$} 
			child { node[one] {$a$} }
			child { node[one, shape=circle, draw=purple, dashed] (g) {$[b]$} 
				child { node[one] {$b$} } } } };
\draw[->, dashed, draw=blue] (a)--(b);
\node[below right =-1em and 0em of a] (c) {\tiny \textcolor{blue}{next}};
\node[] (d) at (c.south) {\tiny \textcolor{blue}{sibling}};
\node[] (e) at (d.south) {\tiny \textcolor{blue}{relation}};
\draw[->, draw=red] (h)--(f);
\node[below right =0.1em and -4em of f] (j) {\tiny \textcolor{red}{parent}};
\node[] at (j.south) {\tiny \textcolor{red}{relation}};
\node[above right =-0.2em and 0em of g] (k) {\tiny \textcolor{purple}{predicate}};
\node[] at (k.south) {\tiny \textcolor{purple}{$\in B^*$}};
\end{tikzpicture}
\end{center}	
\caption{\label{fig:parse-tree} The parse tree of a nested list.}
\end{figure}
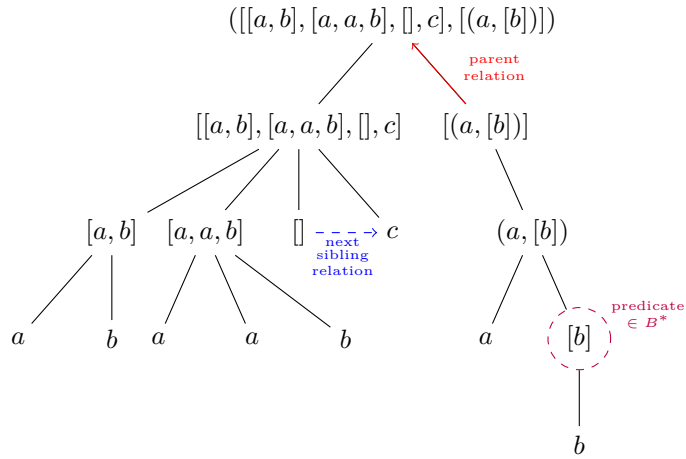

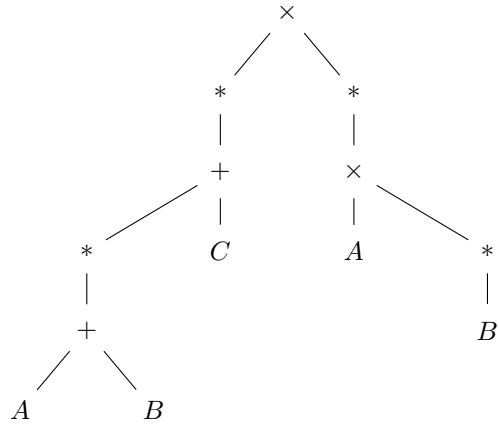
\begin{figure}[hbt]
\begin{center}
\begin{tikzpicture}[sibling distance=5em, level distance=3em]
\tikzstyle{one}=[scale=1, shape=circle]
\node[one] {$\times$}
    child { node[one] {$*$} 
    	child { node[one] {$+$} 
    		child { node[one] {$*$} 
    			child { node[one] {$+$}
    				child { node[one] {$A$} }
    				child { node[one] {$B$} } } } 
    		child { node[one] {$C$} }
    		child { node[one] {} edge from parent[draw=none] } } }
	child { node[one] {$*$} 
		child { node[one] {$\times$}
			child { node[one] {} edge from parent[draw=none] }
			child { node[one] {$A$} }
			child { node[one] {$*$}
				child { node[one] {$B$} } } } };
\end{tikzpicture}
\end{center}	
\caption{\label{fig:syntax-tree} The parse tree of a type in $\types$.}
\end{figure}
 
We write $\underline \Sigma$ for the relational vocabulary used in the structure $\underline x$. This vocabulary has two binary relations, as described in items 2 and 3, as well as one unary relation for every node in the parse tree of the type  $\Sigma$.

\begin{definition}
\label{def:fo-transduction-definable}
Let $\Sigma,\Gamma \in \types$. We say that a function $f : \Sigma \to \Gamma$ is \emph{definable by an \FO-transduction} if it is an \FO-transduction under the encoding $x \mapsto \underline x$; more formally, if there is some \FO-transduction $\varphi$ which makes the following diagram commutes:
\begin{align*}
 \xymatrix@C=4cm{ 
 \Sigma \ar[d]_f \ar[r]^{x \mapsto \underline x} & \text{structures over } {\underline \Sigma} \ar[d]^\varphi \\
 \Gamma \ar[r]^{x \mapsto \underline x} & \text{structures over }{\underline \Gamma}
 }
\end{align*}
\end{definition}

It is important that the encoding $x \mapsto \underline x$ gives the transitive closure of the next sibling relation. For example when the type $\Sigma$ is  $\set{a,b}^*$, our representation  allows a first-order transduction to access the order $<$ on positions, and not just the successor relation. For first-order logic (unlike for  \MSO) there is a significant difference between having access to order vs successor on list positions. 

\subsection{Main result}
\label{subsection:mainresult}
Below is one of the main contributions of this paper.

\begin{theorem}
\label{thm:transductions}
Let $\Gamma,\Sigma \in \types$. A function $f : \Sigma \to \Gamma$ 
is a first-order list function if and only if
it is definable by  an \FO-transduction.
\end{theorem}

Before proving Theorem~\ref{thm:transductions}, let us note the following corollary:
the equivalence of first-order list functions (i.e.~do they give the same output for every input) is decidable. Indeed, we will see below that any first-order list function can be encoded into a string-to-string first-order list function. Using this encoding and Theorem~\ref{thm:transductions}, the equivalence problem of first-order list functions boils down to deciding equivalence of string-to-string \FO-transductions; which is decidable~\cite{doi:10.1137/0211035}.

\medskip

\paragraph*{Proof of the left-to-right implication of Theorem~\ref{thm:transductions}.} The proof of the left-to-right implication of Theorem~\ref{thm:transductions} is by induction following the definition of first-order list functions. The basic functions $\mathsf{projection}$, $\mathsf{co-projection}$ and $\mathsf{distribute}$ are clearly definable by 
\FO-transductions. We prove now that $\mathsf{reverse}$, $\mathsf{flat}$, $\mathsf{append}$, $\mathsf{co-append}$ and $\mathsf{block}$ are also defined by
\FO-transduction.

In all 
of the following cases, let $\rooot(z)$ be a macro for the root node ($\rooot(z)=\neg \exists z' \pare(z',z)$).
In all cases below, for a list $\tau$, let the structure of $\tau$ be $\underline{\tau}$, and  
recall that we have $\pare, \nsib, \sib \in \underline{\Delta}$, the vocabulary of $\underline{\tau}$. 

\begin{enumerate}
\item \textit{$\mathsf{Reverse}$.} Given a list $\tau$, the first-order list function $\mathsf{reverse}(\tau)$ can be implemented using an \FO-transduction as follows: The nodes in the parse tree of $\mathsf{reverse}(\tau)$ is specified by the universe formula
$\varphi^1(x)=true$, selecting all the nodes from the parse tree of $\tau$. 
 The parent-child relations are   
left unchanged but the next sibling relations are reversed.  
\item \textit{$\mathsf{Append}$.} Given $\tau=(x_0, [x_1, \dots, x_n]) \in \Sigma \times \Sigma^*$, the first-order list function  
$\mathsf{append}(\tau)$ resulting in $[x_0, x_1, \dots, x_n]$ is implemented using an  \FO-transduction as follows: 
\begin{enumerate}
\item The nodes in the parse tree of $\mathsf{append}(\tau)$ is specified by the universe formula which 
selects all nodes $y$ in the parse tree of $\tau$ which are not the second child of the root. Remind that the second child of the root here represents the entire list $[x_1, \dots, x_n]$.
$$\varphi^1(y)=\neg [\pare(x,y) \wedge \rooot(x) \wedge \neg \exists z \nsib(y,z)]$$
Note that in the parse tree of $\tau$, the root node has two children, and 
in the constructed parse tree, we omit this second child. 
\item We now specify the parent child relation. This retains the leftmost child of the root in the tree of $\tau$ as a child 
in $\mathsf{append}(\tau)$, and in addition, adds all the children of the second child of 
the root in $\tau$ as the children of the root in $\mathsf{append}(\tau)$. 
\begin{align*}
\varphi^{1,1}(x,y)=\{\neg \rooot(x) \wedge \pare(x,y)\} \vee 
[\rooot(x) \wedge \pare(x,y) \wedge \neg \exists z \nsib(z,y)] \\
\vee 
[\rooot(x) \wedge \{\exists z'[\pare(x,z') \wedge \exists z''(\nsib(z'',z)) \wedge 
\pare(z',y)]\}]
\end{align*}

\end{enumerate}
\item \textit{$\mathsf{Co-append}$.} Given $\tau=[x_0,x_1, \dots, x_n] \in \Sigma^*$, the first-order list function  
$\mathsf{co-append}(\tau)$ resulting in $(x_0, [x_1, \dots, x_n])$ if $n \geq 1$ and undefined otherwise, 
is implemented using an  \FO-transduction as follows: To obtain the parse tree of $\mathsf{co-append}(\tau)$, we make two copies 
of the parse tree of $\tau$.
\begin{enumerate}
\item  The nodes in the first copy are specified by the formula 
$\varphi^1(y)=true$ selecting all the nodes  of $\tau$. 
The second child of the root 
in the parse tree of $\tau$ is selected in the second copy. Note that 
the existence of a second child checks the condition that $n \geq 1$, without which 
the function is not defined.  The formula 

$$\varphi^2(y)=\exists x[\rooot(x) \wedge \pare(x,y) \wedge \exists z. [\pare(x,z) \wedge \nsib(z,y) \wedge \neg \exists z' \nsib(z',z)]
]$$ 

selects the second child of the root in the parse tree of $\tau$. 
\item The parent child relation in the first copy is defined as follows: 
It allows all the edges already present except the parent-child relation
between the root and the nodes which are not the first child of the root.

$$\varphi^{1,1}(x,y)=\psi_1 \vee \psi_2$$ 
where 
$$\psi_1=\rooot(x) \wedge \pare(x,y) \wedge \neg \exists z[\nsib(z,y) \wedge \pare(x,z)]$$
$$\psi_2=\neg \rooot(x) \wedge \pare(x,y)$$ 

The parent child relation in the second copy is defined by
$\varphi^{2,2}(x,y)=false$, since there is a unique node in the second copy.

\item There is one edge from the first copy to the second which makes the root of the first copy the parent of the unique node in the second copy.
The unique node in the second copy is a parent to all the non-leftmost children 
of the root. This is given by formulae 
$\varphi^{1,2}(x,y)=\rooot(x)$ and 
$$\varphi^{2,1}(x,y)=\exists  z[\rooot(z) \wedge \pare(z,y) \wedge \exists z'[\pare(z,z') \wedge \nsib(z',y)]]$$ 
See Figure \ref{figcoapp} where this is illustrated on an example. 
\end{enumerate}
\begin{figure}[h]
\includegraphics[scale=0.3, page=2]{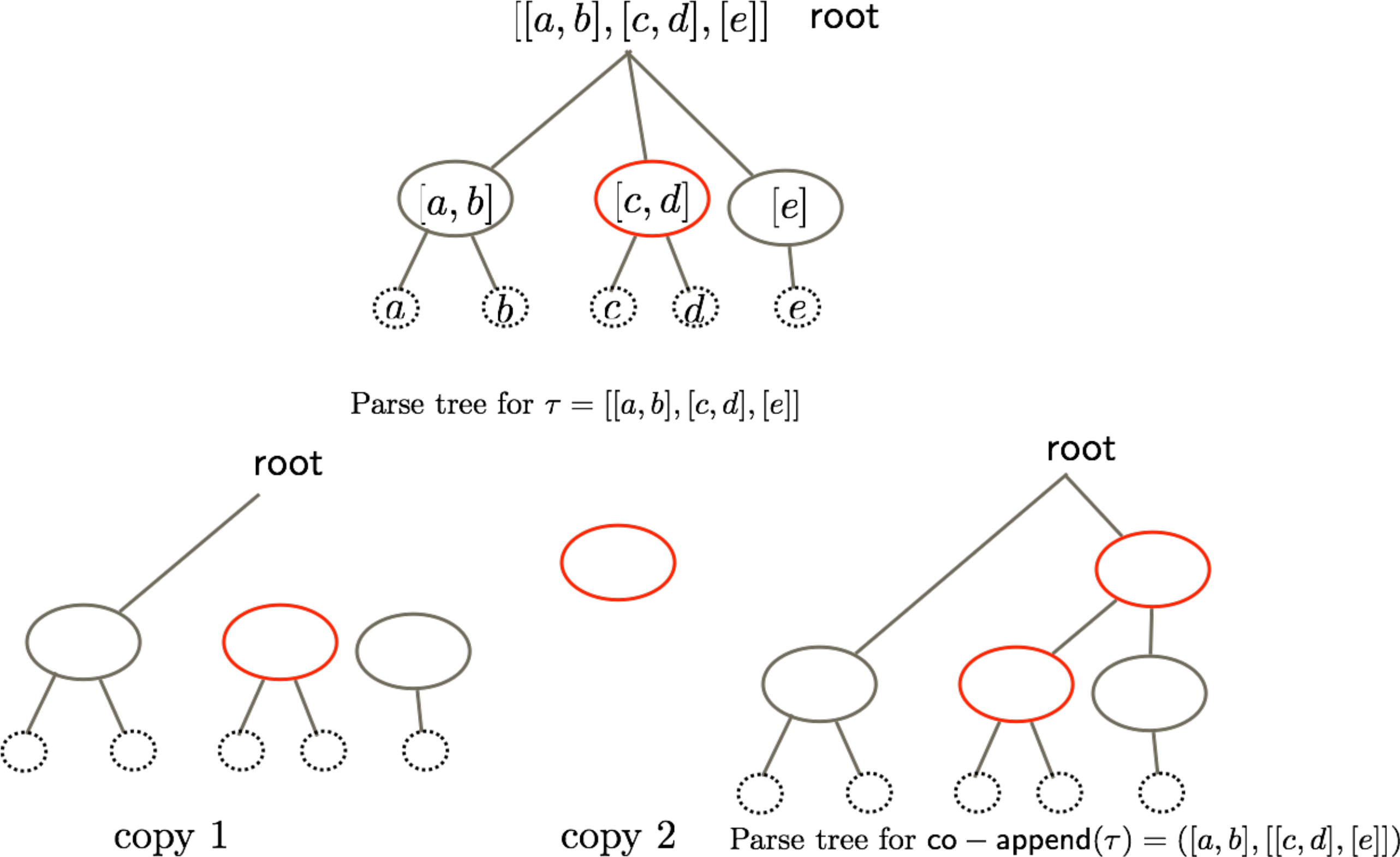}	
\caption{We start with $\tau=[[a,b],[c,d],[e]] \in \Sigma^{**}$.  
On the left is the parse tree of $\tau$. 
 Copy 1 has all the nodes 
 in the parse tree of $\tau$, while copy 2 only has the red circled node
  from the parse tree of $\tau$. The edges in copy 1 include all original edges except the one from the root 
  in the parse tree of $\tau$ to children having a left sibling.  
    The parse tree of $\mathsf{co-append}(\tau)$ is obtained 
 by drawing edges from the root in copy 1 to the only node in copy 2, 
 and from the node in copy 2 to all non-leftmost children of the root,  
 as illustrated.  
}
\label{figcoapp}
\end{figure}

\item 	
\textit{$\mathsf{Flat}$.}
Given a list $\tau$, the first-order list function $\mathsf{flat}(\tau)$ can be implemented using an \FO-transduction as follows:
\begin{enumerate}
\item The nodes in the parse tree of $\mathsf{flat}(\tau)$ 
are all nodes $y$ which are 
not the children of the root node. Note that we do not need any copying of the input structure 
here. It is given by the formula
$$\varphi^1(y)=[\pare(x,y) \wedge \neg \rooot(x)] \vee \rooot(y)$$ 
\item To specify the parent child relation, all nodes other than the root 
have the same parent child relation as before. We also connect the grandchildren of the root 
to the root. This is specified by  
$$\varphi^{1,1}(x,y)=\{\neg \rooot(x) \wedge \pare(x,y)\} \vee 
[\rooot(x) \wedge \{\exists z'[\pare(x,z') \wedge \pare(z',y)]\}]$$
 which says that a non-root node has the same children as before, while 
 the root's children in $\mathsf{flat}(\tau)$ are its grandchildren in $\tau$.  
\end{enumerate}

\item \textit{$\mathsf{Block}$.}
Let's now consider the function $\mathsf{block}$. 
Let $\Sigma$ and $\Gamma$ be types and $\Delta= (\Sigma + \Gamma)^*$. 
Let $\tau \in \Delta$.
Let $\underline{\tau}$ be the relational structure for $\tau$. To obtain 
$\mathsf{block}(\tau)$ as an \FO-transduction, we make two copies 
of the parse tree of $\tau$. The first copy has all the nodes. 
All edges except those defining the parent-child relation 
between the root and its children are present in the first copy. 
The second copy consists of nodes 
which are the children of the root, and whose next sibling 
is of a different type.  The only exception is when dealing with 
the last child of the root, which is always added. 
There are no edge relations between nodes in the second copy. 
We add a parent-child relation between the root of the first copy and all nodes 
in the second copy. Likewise, add a parent-child relation between 
a node $\alpha$ in the second copy with node $\alpha$ in the first copy and all
the following siblings of $\alpha$ who have the same type as $\alpha$. 
\begin{enumerate}
\item The universe formula describing the nodes in the first copy is given by 
$\varphi^1(x)=true$, including all the nodes. 
\item The edges between nodes in the first copy is given by 
$\varphi^{1,1}(x,y)= \neg \rooot(x) \wedge \pare(x,y)$ which 
retains all parent-child relations other than between the root and its children. 
\item Assume that we have finitely many $\type$ predicates $\type_1, \dots, \type_n$ in the structure $\underline{\tau}$. That is, 
$\type_1, \dots, \type_n$, $\pare, \nsib, \sib \in \underline{\Delta}$, the vocabulary of $\underline{\tau}$. 
The universe formula describing nodes in the second copy 
is given by $$\varphi^2(x)=\exists z.[\pare(z,x) \wedge \rooot(z)] \wedge 
\{\bigvee_{i=1}^n [\type_i(x) \wedge \nsib(x,y) \rightarrow \neg \type_i(y)]\}$$
This selects all children of the root which either does not have a next sibling (last child)
or whose next sibling has a different type.
\item The edge relation between nodes in the second copy is $\varphi^{2,2}(x,y)=false$, thereby disallowing any edges. 
\item  The edges from nodes in the first copy to the second copy is given by 
$$\varphi^{1,2}(x,y)=\rooot(x) \wedge \pare(x,y)$$ which 
enables edges from the root of the first copy to all nodes in the second copy. Recall that $\pare(x,y)$ 
is true in $\tau$ for all nodes $y$ in the second copy and the root $x$. 
\item 
The edges from  nodes in the second copy to nodes in the first copy 
is given by 
$$\varphi^{2,1}(x,y)=\psi_1 \wedge \psi_2$$ 
where 
$$\psi_1=[\sib(y,x) \vee (x=y)] \wedge [x \neq y \rightarrow \bigvee_{i=1}^n[\type_i(y) \leftrightarrow \type_i(x)]]$$  
$$\psi_2=
 \neg \exists z' (\sib(z',x) \wedge \sib(y,z') \wedge 
\bigvee_{i=1}^n[\type_i(y) \leftrightarrow \neg \type_i(z')])
$$ 
$\psi_1$ collects all the siblings to the left 
of $x$ (and itself) which have the same  type as $x$, while 
$\psi_2$ ensures that 
the chosen nodes are contiguous and of the same type. 
This ensures that we ``block'' nodes of the same type and assign it one parent, 
and a change of type results in a different parent. Figure \ref{figblock} illustrates this on an example. 
\end{enumerate}

\begin{figure}[h]
\includegraphics[scale=0.3,page=1]{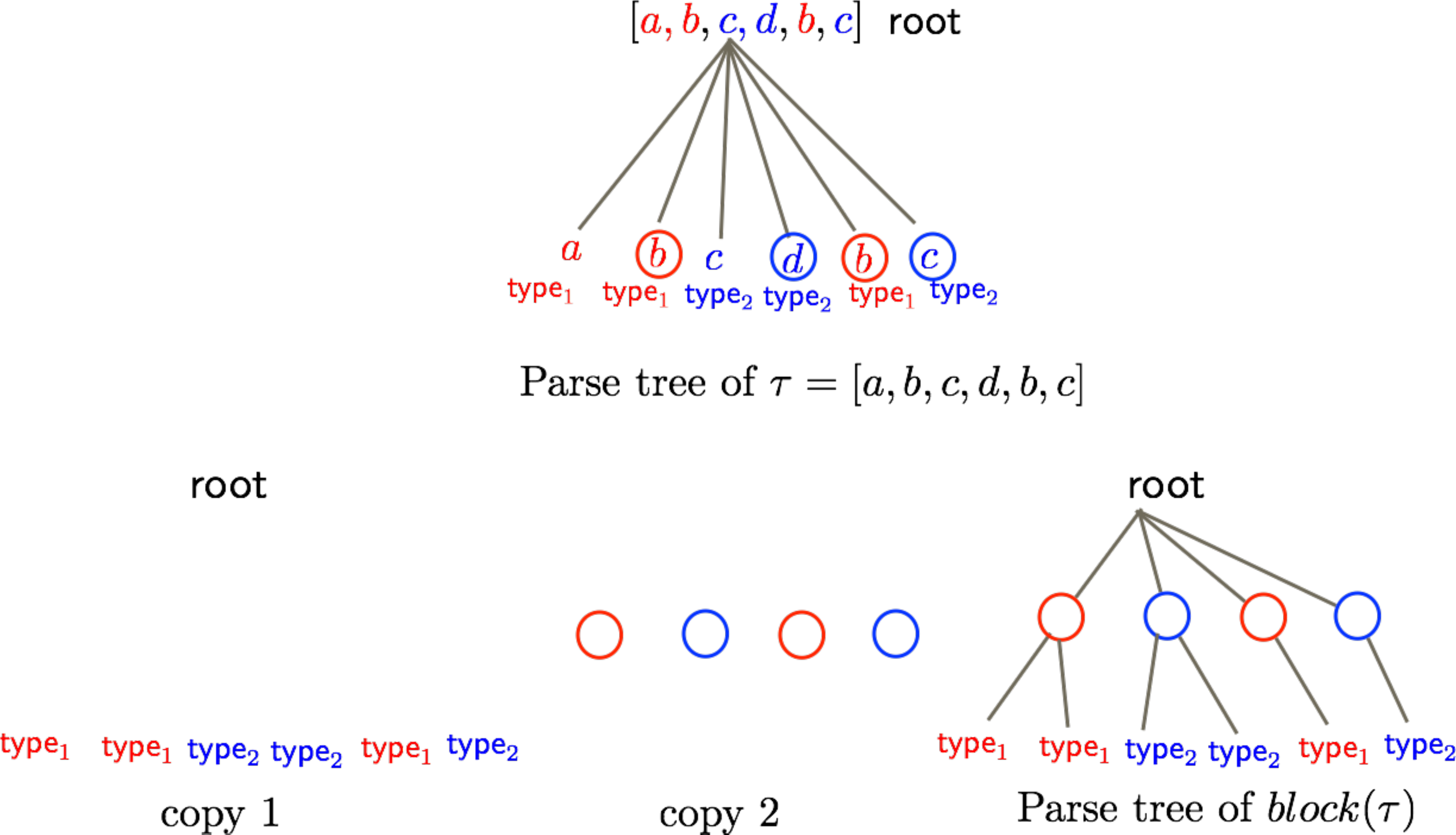}	
\caption{We start with $\tau=[a,b,c,d,b,c] \in (\Sigma +\Gamma)^*$ where $a,b \in \Sigma$ and $c, d \in \Gamma$. 
On the left is the parse tree of $\tau$. The type $\mathsf{type}_1$ 
 represents $\Sigma$ while $\mathsf{type}_2$ represents $\Gamma$. Copy 1 has all the nodes 
 in the parse tree of $\tau$, while copy 2 only has the circled nodes 
 from the parse tree of $\tau$. The parse tree of $\mathsf{block}(\tau)$ is obtained 
 by drawing edges from copy 1 to copy 2, and back as illustrated.  
}
\label{figblock}
\end{figure}

\item Putting together all the basic list functions using combinators : 	Finally, we are left to prove that \FO-transductions are closed under the four combinators. It is clear that \FO-transductions are closed under disjoint union and composition. Moreover, map and pairing can be handled the same way. Consider two \FO-transductions $f$ and $g$ and a relational structure representing a list or a pair. Map of $f$ (resp. pairing $f$ and $g$) is defined by applying $f$ to all the subtrees of the root in the relational structure (resp. applying $f$ to the left subtree of the root and $g$ to the right one). It is clear that this can be done with an \FO-transduction making a number of copies equal to the maximum of the number of copies required for $f$ and $g$ and linking the roots of the subtrees obtained by applying $f$ and $g$ to one unique new root.
\end{enumerate}

\paragraph*{Proof of the right-to-left implication  of Theorem~\ref{thm:transductions}.}
The more challenging right-to-left implication is described in the rest of this section. First, by using an encoding of nested lists of bounded depth via strings, e.g.~{\sc xml} encoding (both the encoding and decoding are easily seen to be both first-order list functions  and definable by \FO-transductions), we obtain the following lemma:

\begin{lemma}
\label{lem:strings-wlog}
To prove the right-to-left implication of Theorem~\ref{thm:transductions}, it suffices to show it for string-to-string functions, i.e.~those of type  $\Sigma^* \to \Gamma^*$ for some finite sets $\Sigma,\Gamma$.
\end{lemma}

Without loss of generality, we can thus only consider the case when both the input and output are strings over a finite alphabet. This way one can use standard results on string-to-string transductions (doing away with the need to reprove the  mild generalisations to nested lists of bounded depth).

Every string-to-string \FO-transduction can be decomposed as a two step process:  (a) apply an aperiodic rational transduction to transform the input word into a sequence of operations which manipulate a fixed number of registers that store words; and then (b) execute the sequence of operations produced in the previous step, yielding an output word. Therefore, to prove Theorem~\ref{thm:transductions}, it suffices to show that (a) and (b) can be done by first-order list functions. Step (a) is Theorem~\ref{thm:sequential-functions}. Step (b) is described in Section~\ref{subsection:sst}. Before 
tackling the proofs, we need to generalise slightly the definition of first-order list functions.

\subsection{Generalised first-order list functions}

In some constructions below, it will be convenient to work with a less strict type discipline, which allows types such as ``lists of length at least three'' or 
\begin{align*}
\set{[x_1,\ldots,x_n] \in \set{a,b}^* : \text{every two consecutive elements differ}}
\end{align*}

\begin{definition}[First-order definable set]
\label{def:typesfo}
Let $\Sigma \in \types$. A subset $P \subseteq \Sigma$ is called \emph{first-order definable} if its characteristic  function $\Sigma \to \set{0,1}$ is a first-order list function.  Let $\typesfo$ denote first-order definable subsets of types in $\types$. 
\end{definition}
When $\Sigma$ is of the form $\Gamma^*$ for some finite alphabet, the above notion coincides with the usual notion of first-order definable language, as in the Sch\"utzenberger-McNaughton-Papert Theorem. The \emph{generalised first-order list functions} are defined to be first-order list functions as defined previously where the domains and co-domains are in $\typesfo$.

\begin{definition}[Generalised first-order list functions]
\label{def:folf}
A function 
$f : \Sigma \to \Gamma \quad \mbox{with} \quad \Sigma, \Gamma \in \typesfo$
is said to be a \emph{generalised first-order list function} if it is obtained by taking some first-order list function (definition~\ref{def:lf}) and restricting its domain and co-domain to first-order definable subsets so that it remains a total function.
\end{definition}

\subsection{Registers}
\label{subsection:sst}

To complete the proof of Theorem~\ref{thm:transductions}, it will be convenient to use a characterisation of \FO-transductions which uses registers, in the spirit of streaming string transducers~\cite{Alur2011}.

\paragraph*{Registers and their updates.}
Let $M$ be a monoid, not necessarily finite, and let $k \in \set{1,2,\ldots}$. We define a \emph{$k$-register valuation} to be a tuple in $M^k$, which we interpret as a valuation of registers called $\set{1,\ldots,k}$ by elements of $M$.  Define a \emph{$k$-register update} over $M$ to be a  parallel substitution, which transforms one $k$-valuation into another using concatenation, as in the following picture.
\mypic{4}
Formally, a $k$-register update is a $k$-tuple of words over $M \cup \set{1,\ldots,k}$. In particular, if $M$ is in $\typesfo$  then also the set of $k$-register updates is in $\typesfo$, and therefore it is meaningful to talk about (generalised) first-order list functions that input and output $k$-register valuations. If $\eta$ is a $k$-register update, we use the name \emph{$i$-th right hand side} for the $i$-th coordinate of the $k$-tuple $\eta$. 

There is a natural right action of register updates on register valuations: if $v \in M^k$ is a $k$-register valuation, and $\eta$ is a $k$-register update, then we define $v \eta \in M^k$ to be the $k$-register valuation where register $i$ stores the value in the monoid $M$ obtained by taking the $i$-th right hand side of $\eta$, substituting each register name $j$ with its value in $v$, and then taking the product in the monoid $M$. This right action can be implemented by a generalised first-order list function, as stated in the following lemma, which is easily proved by inlining the definitions.

\begin{lemma}
\label{lem:right-action}
Let $k$ be a non-negative integer, let $v \in M^k$ and assume that $M$ is a monoid whose universe is in $\typesfo$ and whose product operation is a generalised first-order list function. Then the function which maps a $k$-register update $\eta$ to the $k$-register valuation $v\eta \in M^k$ is a generalised first-order list function.
\end{lemma}

\begin{pr}
Consider a $k$-register update encoded as a $k$-tuple of words, each of which are encoded by a list of elements from $M$ and from $\{1,\ldots, k\}$. First, because $k$ is fixed and using projection and pairing, we can only consider the case of a single word. Let $v$ as in the lemma. Because $v$ is fixed, the function from the finite set $\{1,\ldots, k\}$ associating with $i$ the $i$th component of $v$ is a first-order list function. The function identity over $M$ is also a generalised first-order list function, so by using disjoint union first and then map, the function which transforms a list of elements from $M + \{1,\ldots, k\}$ into a list of element from $M$ by replacing $i$ by the $i$th component of $v$ is a generalised first-order list function. Finally, by using the product operation in $M$ which is a generalised first-order list function, one can compute the product of the elements of the list.
\end{pr}

\paragraph*{Non duplicating monotone updates.} 
A $k$-register update is called \emph{nonduplicating} if each register appears at most once in the concatenation of all the right hand sides; and it is called \emph{monotone} if after concatenating the right hand sides (from $1$ to $k$), the registers appear in strictly increasing order  (possibly with some registers missing). We write $\regmon M k$ for the set of nonduplicating and monotone $k$-register updates. 	

\bigskip

Lemma~\ref{lem:decomposition-of-fo-transduction} says that every string-to-string \FO-transduction can be decomposed as follows: (a) apply an aperiodic rational transduction to compute a sequence of monotone nonduplicating register updates; then (b) apply all those register updates to the empty register valuation (which we denote by $\bar \varepsilon$ assuming that the number of registers $k$ is implicit), and finally return the value of the first register. 

\begin{lemma}
\label{lem:decomposition-of-fo-transduction}
Let $\Sigma$ and $\Gamma$ be finite alphabets. Every \FO-transduction $f : \Sigma^* \to \Gamma^*$ can be decomposed as:
\begin{align*}
\xymatrix@C=3cm{\Sigma^* \ar[d]^{g} \ar[r]^f     &\Gamma^* \\   \Delta^*    \ar[r]_{\text{apply updates to $\bar \varepsilon$}} & (\Gamma^*)^k \ar[u]_{\mathsf{projection}_1}}
\end{align*}
for some positive integer $k$, where $\Delta$ is a finite subset of $\regmon {(\Gamma^*)} k$ (i.e.~a finite set of $k$-register updates that are monotone and nonduplicating), $g : \Sigma^* \to \Delta^*$ is an aperiodic rational function and $\mathsf{projection}_1$ is the projection of a k-tuple of $(\Gamma^*)^k$ on its first component.
\end{lemma}

\begin{pr}
We use an equivalent characterisation of \FO-transduction in terms of streaming 
string transducers (SST) (first shown for \MSO-transduction in~\cite{AlurC10}). By~\cite{FiliotKT14}, an \FO-transduction 
$f:\Sigma^* \to \Gamma^*$ can be computed by an SST $\mathbf{F}$ whose 
register updates are nonduplicating and whose transition monoid is aperiodic.
A possible definition for the transition monoid of an SST 
is given in~\cite{DartoisJR16} (called \textit{substitution transition monoid}).
Elements of the monoid are functions mapping a state $p$ of the SST to a pair $(q,r)$ formed with a state $q$ of the SST and a $k$-register update $r$ containing
only names of registers (and no element of $\Gamma$). Every word $u$ is mapped 
to such a function $g_u$ such that $g_u(p) = (q,r)$ if there is a run in the SST
on $u$ from state $p$ to state $q$ where the registers are updated according 
to $r$ with possibly elements of $\Gamma$ inserted in the products. 

We will first transform $\mathbf{F}$ to add a regular look-ahead which will guess
which registers are output at the end of the computation on a given word and 
in which order.
More precisely, there is a rational function $\tilde{g}$
taking as input a state $p$ of $\mathbf{F}$ and a word $v$ and which outputs the sequence of registers $\tilde{g}(p,v) = (r_1, r_2, \dotsm, r_\ell)$ which appears (respecting the order) in the update of the register output in $q$ in $r$ where $g_v(p) = (q,r)$.
Note that this sequence contains at most $k$ registers, all distinct.
Because the transition monoid of $\mathbf{F}$ is aperiodic, so is $\tilde{g}$.

Now, let $\mathbf{G}$ be the SST constructed from $\mathbf{F}$ and $\tilde{g}$
as follows:

States are pairs of a state $p$ of $\mathbf{F}$ and a sequence of at most $k$ 
distinct registers which corresponds to $\tilde{g}(p,v)$ for some word $v$. 
Transitions are similar as the ones in $\mathbf{F}$, such that:
\begin{itemize}
\item from a state $(p,\tilde{g}(p,uv))$, one can reach the state $(q,\tilde{g}(q,v))$ when reading $u$, where $q$ is the state reached from $p$ by 
reading $u$ in $\mathbf{F}$, 
\item the register updates are modified in such a way that the registers are reordered according to $\tilde{g}(p,v)$ so as to maintain monotone updates all along the computation.
\end{itemize}
We can moreover assume that the output register is always the first one.

The set $\Delta$ is now defined as the finite set of register updates on the transitions in $\mathbf{G}$, which are nonduplicating and monotone. The function 
$g$ is the function mapping a word $u$ to the 
sequence of register updates performed while reading $u$ in $\mathbf{G}$
from the initial state to the state $(q,\tilde{g}(q,\varepsilon))$, where 
$q$ is the state reached in $\mathbb{F}$ by reading $u$ from the initial state. 
The function $g$ is thus aperiodic rational.

Because $\mathbf{F}$ and $\mathbf{G}$ are equivalent, 
the function $f$ can be decomposed 
into the three functions as in the statement of the lemma.
\end{pr}

Thanks to Lemma~\ref{lem:strings-wlog} and the closure of first-order list function under composition, it is now sufficient to prove that the bottom three functions of Lemma~\ref{lem:decomposition-of-fo-transduction} are first-order list functions, in order to complete the proof of Theorem~\ref{thm:transductions}. This is the case of the function $\mathsf{projection}_1$ by definition and of the aperiodic rational function $g$ by Theorem~\ref{thm:sequential-functions}. We are thus left to prove the following lemma, which is the subject of Section~\ref{section:register}.

\begin{lemma}
\label{lem:product-operations-in-lf}
Let $\Gamma \in \types$, let $k$ be a positive integer and let $\Delta$ a finite subset of $\regmon {(\Gamma^*)} k$. The function from $\Delta^*$ to $(\Gamma^*)^k$ which maps a list of nonduplicating monotone $k$-register updates to the valuation obtained by applying these updates to the empty register valuation is a first-order list function.
\end{lemma}

\section{The register update monoid}
\label{section:register}
The goal of this section is to prove Lemma~\ref{lem:product-operations-in-lf}.
We will prove a stronger result which also works for a monoid other than $\Gamma^*$, provided that its universe is a first-order definable set and its product operation is a generalised first-order list function. This result is obtained as a corollary of Theorem~\ref{thm:regmon} below. In order to state this theorem formally, we need to view the product operation: $(\regmon M k)^* \to \regmon M k$ as a generalised first-order list function. The domain of the above operation is  in  $\typesfo$ from Definition~\ref{def:typesfo}, since being monotone and nonduplicating are first-order definable properties. 

\begin{theorem}
\label{thm:regmon} 
Let $M$ be a monoid whose universe is in $\typesfo$  and whose product operation is a generalised first-order list function. Then the same is true for $\regmon M k$, for every $k \in \set{0,1,\ldots}$. 
\end{theorem}

Lemma~\ref{lem:product-operations-in-lf} follows from the above theorem applied to $M = \Gamma^*$, and from Lemma~\ref{lem:right-action}. Indeed, given $\Gamma \in \types$, the universe of $\regmon {\Gamma^*} k$ is in $\typesfo$ and its product operation is a generalised first-order list function by Theorem~\ref{thm:regmon}. The function from Lemma~\ref{lem:product-operations-in-lf} is then the composition of the product operation in $\Delta^*$ (which corresponds to the product operation in $\regmon {\Gamma^*} k$), which transforms a list of updates from $\Delta$ into an update of $\regmon {\Gamma^*} k$, and the evaluation of this update on the empty register valuation. This last function is a generalised first-order list function by Lemma~\ref{lem:right-action} with $v=\bar \varepsilon$. 
Implicitly, we use the fact that the right action is compatible with the monoid structure of $\regmon M k$, i.e.
\begin{align*}
  v (\eta_1 \eta_2) = (v \eta_1) \eta_2 \qquad \mbox{for $v \in M^k$ and $\eta_1,\eta_2 \in \regmon M k$}.
\end{align*}
Summing up, we have proved that the function of type $\Delta^* \to (\Gamma^*)^k$ discussed in Lemma~\ref{lem:product-operations-in-lf} is a generalised first-order list function. Since its domain and co-domain are in $\types$, it is also a first-order list function. This completes the proof of Lemma~\ref{lem:product-operations-in-lf}.

It remains to prove Theorem~\ref{thm:regmon}. 
We do this using factorisation forests, with our proof strategy encapsulated in the following lemma, using the notion of \emph{homogeneous lists}: a list $[x_1,\ldots,x_n]$ is said to be homogeneous under a function $h$ if $h(x_1)= \cdots = h(x_n)$.

\begin{lemma}
\label{lem:simon-strategy}
Let $P$ be a monoid whose universe is in $\typesfo$. The following conditions are sufficient for the product operation to be a generalised first-order list function:
\begin{enumerate}
	\item the binary product $P \times P \to P$ is a generalised first-order list function; and
	\item there is a monoid homomorphism $h : P \to T$, with $T$ a finite aperiodic monoid, and a generalised first-order list function $P^* \to P$ that agrees with the product operation of $P$ on all lists that are homogeneous under~$h$.
\end{enumerate}
\end{lemma}

\begin{pr}
Our goal is to compute the product of a list $x \in P^*$. Consider $h$ and $T$ as given in  condition \textbf{2.} and compute an $h$-factorisation of $x$ with a first-order list function using Theorem~\ref{thm:simon-compute}. The depth of such a tree is bounded by a constant depending only on $T$. Then, by induction on the depth, one can prove that there is a generalised first-order list function computing the product of the labels of the leaves, using condition \textbf{1.} to deal with nodes of degree $2$ and condition \textbf{2.} to deal with nodes of degree at least $3$. By composition, we get that the product operation of $P^*$ is a generalised first-order list function.
\end{pr}

In order to prove Theorem~\ref{thm:regmon}, it suffices to show that if a monoid $M$ satisfies the assumptions of Theorem~\ref{thm:regmon}, then the monoid $\regmon M k$ satisfies conditions 1 and 2 in Lemma~\ref{lem:simon-strategy}. Let us fix for the rest of this section a monoid $M$ which satisfies the assumptions of Theorem~\ref{thm:regmon}, i.e.~its  universe is in $\typesfo$ and its product operation is a generalised first-order list function. Condition 1 of Lemma~\ref{lem:simon-strategy} for $\regmon M k$ is easy.
Indeed, consider two elements of $\regmon M k$, say $u=[u_1,u_2,\ldots,u_k]$ and $v=[v_1,v_2,\ldots,v_k]$. The $u_i$'s and $v_i$'s are lists of elements in $M\cup\{1,\ldots,k\}$. To obtain the product, we need to replace every occurrence of $j \in \{1,\ldots,k\}$ in the $v_i$'s by $u_j$. Because $\{1,\ldots,k\}$ is finite, using the if-then-else construction, one can prove that the function from $\{1,\ldots,k\} \times \regmon M k$ associating $(i,u)$ with the register update $\mathsf{projection}_i(u)$ is a generalised first-order list function. 
Then, as a first step, replace every element $s$ of $M$ in the $v_i$'s by the singleton list $[s]$. Then, replace every element $j$ of $\{1,\ldots,k\}$ in the $v_i$'s by the pair $(j,u)$. Finally apply the generalised first-order function, as defined above, to those elements. The desired list is obtained by flattening. 

We  focus now on condition 2, i.e.~showing that the product operation can be computed by a first-order list function, for lists which are homogeneous under some homomorphism into a finite monoid. For this, we need to find the homomorphism $h$. For a $k$-register update $\eta$, define $h(\eta)$, called its  \emph{abstraction}, to be the same as $\eta$, except that  all monoid elements are removed from the right hand sides, as in the following picture:
\mypic{6}
Intuitively, the abstraction only says which registers are moved to which ones, without saying what new monoid elements (in blue in the picture) are created. 
Having the same abstraction is easily seen to be a congruence on $\regmon M k$, and therefore the set of abstractions, call it $T_k$, is itself a finite  monoid, and the abstraction function $h$ is a monoid homomorphism. We say that a list $[x_1,\ldots,x_n]$ in $\regmon M k$ is $\tau$-homogeneous for some $\tau$ in $T_k$ if it is homogeneous under the abstraction $h$ and $\tau = h(x_1)= \cdots = h(x_n)$.
We claim that item 2 of Lemma~\ref{lem:simon-strategy} is satisfied when using the abstraction homomorphism. 

\begin{lemma}
\label{lem:idempotent-product} 
Given a non-negative integer $k$ and $\tau \in T_k$, there is a generalised first-order list function from $(\regmon M k)^*$ to $\regmon M k$ which agrees with the product in the monoid $\regmon M k$ for arguments which are $\tau$-homogeneous.
\end{lemma}

Since there are finitely many abstractions, and a generalised first-order list function can check if a list is $\tau$-homogeneous, the above lemma yields item 2 of Lemma~\ref{lem:simon-strategy} using a case disjunction as described in Example~\ref{ex:if-then-else}. Therefore, proving the above lemma finishes the proof of Theorem~\ref{thm:regmon}, and therefore also of Theorem~\ref{thm:transductions}. The rest of the section is devoted to the proof of Lemma~\ref{lem:idempotent-product}. In Section~\ref{section:one-register}, we prove the special case of Lemma~\ref{lem:idempotent-product} when $k=1$, and in Section~\ref{section:more-registers}, we deal with the general case (by reducing it to the case $k=1$).

\subsection{Proof of Lemma~\ref{lem:idempotent-product}: One register}
\label{section:one-register}

Let us first prove the special case of Lemma~\ref{lem:idempotent-product} when $k=1$. In this case, there are two possible abstractions:
\mypic{18}

The right case is easy to deal with: in that case, the product of a sequence of elements in $(\regmon M k)^*$ which are $\tau$-homogeneous, is equal to the last element of the list. This can be obtained using the first-order list function $\mathsf{last}$ from Example~\ref{ex:headtail}.

Let us consider now the more interesting left case; fix $\tau$ to be the left abstraction above. Here is a picture of a list $[\eta_1,\ldots,\eta_n] \in (\regmon M 1)^\star$ which is $\tau$-homogeneous:
\mypicsmall{5}

Our goal is to compute the product of such a list, using a generalised first-order list function.  For $\eta \in \regmon M 1$ define $\leftsub(\eta)$ (respectively, $\rightsub(\eta)$) to be the  list in $M^*$ of monoid elements that appear in $\eta$ before (respectively, after) register 1.  Here is a picture
\mypic{23}


Using the list comma function from Example~\ref{ex:comma-function}, one can transform $\eta$ by grouping the elements of $M$ in lists, using registers as separators. This way, $\eta$ is transformed into the list of lists $[\leftsub(\eta),\rightsub(\eta)]$. Using $\mathsf{head}$ and $\mathsf{last}$, we get:

\begin{claim}\label{claim:implicit-normalisation}
Both $\leftsub$ and $\rightsub$  are generalised first-order list functions $\regmon M 1 \to M^*$.
\end{claim}

Let $s,t \in M^*$ be the respective flattenings of the lists 
$$[\leftsub(\eta_n),\ldots,\leftsub(\eta_1)] \quad \text{ and } \quad [\rightsub(\eta_1),\ldots,\rightsub(\eta_n)].$$
Note the reverse order in the first list. Here is a picture:
\mypicsmall{14}
The function $[\eta_1,\ldots,\eta_n] \mapsto (s,t)$ is  a generalised first-order list function, using Claim~\ref{claim:implicit-normalisation}, map, reversing, flattening and pairing.
The product of the 1-register valuations $[\eta_1,\ldots,\eta_n]$ is the register valuation where the (only) right hand side is the concatenation of $s,[1],t$. Therefore, this product can be computed by a generalised first-order list function.

This completes the proof of Lemma~\ref{lem:idempotent-product} in the case of $k=1$, in particular we now know that Theorem~\ref{thm:regmon} is true for $k=1$. 

\subsection{Proof of Lemma~\ref{lem:idempotent-product}: More registers}
\label{section:more-registers}

We now prove the general case of Lemma~\ref{lem:idempotent-product}. Let $\tau \in T_k$ be an abstraction. We need to show that a generalised first-order list function can compute the product operation of $\regmon M k$ for inputs that are $\tau$-homogeneous.  Our strategy is to use homogeneity to reduce to the case of one register, which was considered in the previous section.
As a running example (for the proof in this section) we  use the following $\tau$: 
\mypic{24}
Define $G$ to be a directed graph where the vertices are the registers $\set{1,\ldots,k}$ and which contains an edge $i \leftarrow j$ if the $i$-th element of the abstraction $\tau$ contains register $j$, i.e.~the new value of register $i$ after the update uses register $j$. Here is a picture of the graph $G$ for our running example:
\mypic{25}
Every vertex in the graph has outdegree at most one (because $\tau$ is nonduplicating) and the only types of cycles are self-loops (because $\tau$ is monotone). Because registers that are in different weakly connected components do not interact with each other and then can be treated separately, without loss of generality we can assume that $G$ is weakly connected (i.e.~it is connected after forgetting the orientation of the edges).

Consider a $\tau$-homogeneous list $[\eta_1,\ldots,\eta_n]$ of $k$-register updates. Here is a picture for our running example:
\mypicsmall{26}
A register $i \in \set{1,\ldots,k}$ is called \emph{temporary} if it does not have a self-loop in the graph $G$ (we will also say that the vertex is temporary). In our running example, the temporary registers are 1,3 and 4. Because the outdegree of all the vertices in $G$ is at most $1$, if a vertex has an incoming edge from a different vertex in the graph then this latter must be temporary.  
The key observation about temporary registers is that their value depends only on the last $k$ updates, as shown in the following picture: \mypicsmall{27}
Indeed, an incoming edge in a temporary vertex $i$ must come from a different temporary vertex, so the value in $i$ depends only on the values of the temporary registers corresponding to the vertices in simple paths without self-loop reaching $i$, and thus which occur in the last $k$ updates. 

Because the temporary registers depend only on the recent past, the values of temporary registers can be computed using a generalised first-order list function (as formalised in Claim~\ref{claim:only-temporary} below). 

\begin{claim}
\label{claim:only-temporary}
Assume that $\tau \in \absmon k$ is such that all the registers are temporary. Consider the function:
\begin{align*}
\xymatrix@C=2cm{ (\regmon M k)^* \cap \text{$\tau$-homogeneous} \ar[r]^-f &  (\regmon M k)^* }
\end{align*}
which maps an input $[\eta_1,\ldots,\eta_n]$ to the list $[\eta'_1,\ldots,\eta'_n]$ where $\eta'_i$ is equivalent to the product of the prefix $\eta_1 \cdots \eta_i$. Then $f$ is a generalised first-order list function.
\end{claim}

\begin{pr}
If all registers are temporary, then the product of a $\tau$-homogeneous list is the same as the product of its last $k$ elements. Therefore, we can prove the lemma using the window construction from Example~\ref{ex:window} (for a window of size $k$) and the binary product. 
\end{pr}

If the graph $G$ is connected, as supposed, then there is at most one register that is not temporary. For this register, we use the result on one register proved in the previous section. Indeed, let $[\eta_1,\ldots,\eta_n]$ be a $\tau$-homogeneous list and let $r$ be the only register which is not temporary. Each $\eta_i$ is a list of $k$ register updates. Let us denote it by: $\eta_i = [w^i_1,\ldots,w^i_k]$. Recall that $w^i_j \in (M \cup \{1,\ldots k\})^\star$. By using Claim~\ref{claim:only-temporary} and map, there is a generalised first-order list function which computes $[\eta'_1,\ldots,\eta'_n]$ where $\eta'_i = [w'^i_1,\ldots,w'^i_k]$ with $w'^i_j$ equivalent to $j$-th right hand-side in the product $\eta_1 \cdots \eta_i$ if $j \neq r$ and equal to $w^i_r$ if $j=r$. The list $w'^i_j$ for $j\neq r$ depends only on elements of $M$ and possibly on the initial valuation of the registers. Then, the list $[\eta'_1,\ldots,\eta'_n]$ can be treated as a list of updates, using only one register (register $r$), which is solved in the previous section.

\section{Regular list functions}
\label{section:regular}
In Theorem~\ref{thm:transductions}, we have shown that \FO-transductions  are the same as  first-order list functions. In this section, we discuss the \MSO version of the result.

An \MSO-transduction is defined similarly as an \FO-transduction (see Definition~\ref{def:fo-transduction}), except that the interpretations are allowed to use the logic \MSO instead of only first-order logic\footnote{This definition differs slightly from \MSO-transductions as defined in~\cite[Section 1.7]{DBLP:books/daglib/0030804}, because it does not allow guessing a colouring, but for functional transductions on objects from our type system $\types$, the colourings are superfluous and can be pushed into the formulas from the interpretation.}. When restricted to functions of the form $\Sigma^* \to \Gamma^*$ for finite alphabets $\Sigma,\Gamma$, these are exactly the regular string-to-string functions discussed in the introduction. 

To capture \MSO-transductions, we extend the first-order list functions with product operations for finite groups in the following sense.  Let $G$ be a finite group. Define its \emph{prefix multiplication function} to be
\begin{align*}
  [g_1,\cdots,g_n] \in G^*  \qquad \mapsto \qquad [h_1,\ldots,h_n] \in G^*
\end{align*}
where $h_i$ is the product of the list $[g_1,\ldots,g_i]$.  Let the \emph{regular list functions} to be defined the same way as the first-order list functions (Definition~\ref{def:lf}), except that for every finite group $G$, we add its prefix multiplication function to the base functions.

\begin{theorem}
\label{thm:msot}
Given $\Sigma,\Gamma \in \types$ and a function $f : \Sigma \to \Gamma$, the following conditions are equivalent:
	\begin{enumerate}
		\item $f$ is defined by an \MSO-transduction;
		\item $f$ is a regular list function.
	\end{enumerate}
\end{theorem}


\begin{proof}
	The bottom-up implication is straightforward, since the group product operations are seen to be \MSO-transductions (even sequential functions). 
	
	For the top-down implication, we use a number of existing results to break up an \MSO-transduction into smaller pieces which turn out to be regular list functions. 
	By~\cite[Theorem 2]{Colcombet:2007cm} applied to the special case of words (and not trees),  every \MSO-transduction can be decomposed as a composition of (a) a rational function; followed by (b) an \FO-transduction. Since \FO-transductions are contained in regular list functions by Theorem~\ref{thm:sequential-functions}, and regular list functions are closed under composition, it is enough to show that every rational function is a regular list function. By Elgot and Mezei~\cite{Elgot:1965fn}, every rational function can be decomposed as: (a) a sequential function~\cite[Section 2.1]{Filiot:2016iw}; followed by (b) reverse; (c) another sequential function; (d) reverse again. Since regular list functions allow for reverse and composition, it remains to deal with sequential functions. By the Krohn-Rhodes Theorem~\cite[Theorem A.3.1]{Straubing:2012vv}, every sequential function is a composition of sequential functions where the state transformation monoid of the underlying automaton is either aperiodic (in which case we use Theorem~\ref{thm:sequential-functions}) or a group (in which case we use the  prefix multiplication functions for groups).
	
	An alternative approach to proving the top-down implication would be to revisit the proof of Theorem~\ref{thm:transductions}, with the only important change being a group case needed when computing a factorisation forest for a semigroup that is not necessarily aperiodic.
\end{proof}

\section{Conclusion}
\label{section:conclusion}
The main contribution of the paper is to give a characterisation of the regular string-to-string transducers (and their  first-order fragment) in terms of functions on lists, constructed from basic ones, like reversing the order of the list, and closed under combinators like composition. 


One of the principal design goals of our formalism is to be easily extensible. We end the paper with some possibilities of such extensions, which we leave for future work.

One idea is to add new basic types and functions. For example, one could add an infinite atomic type, say the natural numbers $\Nat$, and some functions operating on it, say  the function $\Nat \times \Nat \to \set{0,1}$ testing for equality. Is there a logical characterisation for the functions obtained this way?

\MSO-transductions and \FO-transductions are linear in the sense that the size of the output is linear in the size of the input; and hence our basic functions need to be linear and the combinators need to preserve linear functions. What if we add  basic operations that are  non-linear, e.g.
\begin{align*}
(a,[b_1,\ldots,b_n]) \quad  \mapsto \quad  [(a,b_1),\ldots,(a,b_n)] \end{align*}
which is sometimes known as ``strength''?
A natural candidate for a corresponding logic would  use  interpretations where output positions are interpreted in  pairs (or triples, etc.) of input positions.

Finally, our type system is based on lists, or strings. What about other data types, such as trees, sets, unordered lists, or graphs? Trees seem  particularly tempting, being a fundamental data structure with a  developed transducer theory, see e.g.~\cite{Alur:2017gh}. Lists and the other data types discussed above can be equipped with a monad structure, which seems to play a role in our formalism. Is there anything valuable that can be taken from this paper which works for arbitrary monads?

\bibliographystyle{plain}
\bibliography{bib}
\end{document}